\documentclass[a4paper,USenglish,dvipsnames]{article}
\usepackage[margin=2.5cm]{geometry}
\usepackage[utf8]{inputenc}
\usepackage[textwidth=2.2cm,textsize=scriptsize]{todonotes}
\usepackage[square,numbers,sort]{natbib}

\usepackage{xcolor}
\usepackage{amsmath}
\usepackage{amsthm}
\usepackage{amssymb}
\usepackage{mathtools}
\usepackage{enumerate}
\usepackage{thm-restate}
\usepackage{nicefrac}
\usepackage{hyperref}
\usepackage{cleveref}
\usepackage{wasysym}
\usetikzlibrary{positioning}
\usepackage{algorithm}
\usepackage[noend]{algpseudocode}
\usepackage{caption}
\usepackage{subcaption}

\newcommand{\dens}{d}
\newcommand{\lfrac}[2]{#1/#2}

\newtheorem{theorem}{Theorem}
\newtheorem{lemma}[theorem]{Lemma}
\newtheorem{corollary}[theorem]{Corollary}

\title{Bicriterial Approximation for the Incremental \\ Prize-Collecting Steiner-Tree Problem\thanks{Supported by Deutsche Forschungsgemeinschaft (DFG, German Research Foundation) through subprojects A07 and A09 of CRC/TRR154  and under Germany’s Excellence Strategy -- The Berlin Mathematics Research Center MATH+ (EXC-2046/1, project IDs: 390685689).}}
\author{Yann Disser\thanks{Department of Mathematics, TU Darmstadt, Germany, Email: disser@mathematik.tu-darmstadt.de}, Svenja M.\ Griesbach\thanks{Institute of Mathematics, TU Berlin, Germany, Email: griesbach@math.tu-berlin.de}, Max Klimm\thanks{Institute of Mathematics, TU Berlin, Germany, Email: klimm@math.tu-berlin.de}, Annette Lutz\thanks{Department of Mathematics, TU Darmstadt, Germany, Email: lutz@mathematik.tu-darmstadt.de}}
\date{}

\newcommand{\R}{\mathbb{R}}
\newcommand{\N}{\mathbb{N}}
\newcommand{\T}{\mathcal{T}}
\newcommand{\M}{\mathcal{M}}
\newcommand{\OPT}{\textsc{Opt}}
\newcommand{\ALG}{\textsc{Alg}}

\newcommand{\lrp}{\gamma} 
\newcommand{\contract}[2]{#1/#2}
\DeclareMathOperator*{\argmax}{arg\,max}
\DeclareMathOperator*{\argmin}{arg\,min}
\DeclarePairedDelimiter{\abs}{\lvert}{\rvert}
\newcommand{\overbar}[1]{\mkern 1.5mu\overline{\mkern-1.5mu#1\mkern-1.5mu}\mkern 1.5mu}

\begin{document}
\maketitle

\begin{abstract}
    We consider an incremental variant of the rooted prize-collecting Steiner-tree problem with a growing budget constraint.
    While no incremental solution exists that simultaneously approximates the optimum for all budgets, we show that a bicriterial $(\alpha,\mu)$-approximation is possible, i.e., a solution that with budget $B+\alpha$ for all $B \in \R_{\geq 0}$ is a multiplicative $\mu$-approximation compared to the optimum solution with budget~$B$.
    For the case that the underlying graph is a tree, we present a polynomial-time density-greedy algorithm that computes a $(\chi,1)$-approximation, where $\chi$ denotes the eccentricity of the root vertex in the underlying graph, and show that this is best possible.
    An adaptation of the density-greedy algorithm for general graphs is $(\gamma,2)$-competitive where $\gamma$ is the maximal length of a vertex-disjoint path starting in the root. While this algorithm does not run in polynomial time, it can be adapted to a $(\gamma,3)$-competitive algorithm that runs in polynomial time. 
    We further devise a capacity-scaling algorithm that guarantees a $(3\chi,8)$-approximation and, more generally, a~$\smash{\bigl((4\ell - 1)\chi, \frac{2^{\ell + 2}}{2^{\ell}-1}\bigr)}$-approximation for every fixed~$\ell \in \mathbb{N}$.
\end{abstract}

\newpage

\section{Introduction}

Prize-collecting Steiner-tree problems serve as a model to study the fundamental trade-off between the cost of installing a network and harnessing its benefits in terms of covering vital vertices. 
They have been used to study, e.g.,  expanding telecommunication networks~(Johnson et al.~\cite{JohnsonMP00}) or pipeline networks~(Salman et al.~\cite{SalmanCRS01}). 
Formally, we are given an undirected graph $G = (V,E)$, 
a positive \emph{edge cost} $c(e) \in \mathbb{R}_{> 0}$ for each edge~$e \in E$, 
a non-negative \emph{vertex prize} $p(v) \in \mathbb{R}_{\geq 0}$ for each vertex~$v \in V$, and a specified \emph{root vertex}~$r \in V$.
In the applications above, vertices correspond to geographical locations such as intersections of a road network, the premises of telecommunication customers, or the locations of oil wells.
The prize of a vertex represents the estimated revenue associated with connecting the vertex to the core network, represented by the root.
The cost of an edge corresponds to, e.g., the monetary cost of laying a fiber-optic cable or a pipeline between the respective end vertices.

Let $\T$ denote the set of all subtrees of $G$ containing $r$. Then, the \emph{budget version} of the prize-collecting Steiner-tree problem for a given budget $B \in \R_{\geq 0}$ is the optimization 
\begin{align}
\max \bigl\{p(T) : T \in \T \text{ with }  c(T) \leq B\bigr\},\label{eq:problem}
\end{align}
where for a subgraph $G' = (V',E') \subseteq G$, we write $p(G') \coloneqq \sum_{v \in V'} p(v)$ and $c(G') \coloneqq \sum_{e \in E'} c(e)$. 
This objective is faced, e.g., by any company interested in building the most profitable telecommunication or pipeline network given limited financial funds.
In the following, we arbitrarily fix a rooted subtree~$\OPT(B)$ attaining the maximum in \eqref{eq:problem}.

In the scenarios mentioned above, it is realistic to assume that the network is expanded over time as the company's budget increases. 
This motivates us to
study an incremental version of the prize-collecting Steiner-tree problem where approximately optimal Steiner-trees have to be maintained under a growing budget constraint. 
Formally, an \emph{incremental solution} to an instance of the prize-collecting Steiner-tree problem is given by an ordering $\pi = (e_{\pi(1)},\dots,e_{\pi(l)})$ of a subset of the edges~$E$ such that every prefix~$(e_{\pi(1)},\dots,e_{\pi(j)})$ of~$\pi$ with $j \leq l$ induces a rooted subtree $T_j \subseteq G$ and the rooted subtree~$T_l$ spans all vertices in $V^{\ast} := \{v \in V : p(v) > 0\}$.
In the applications above, an incremental solution corresponds to an order in which to build the fiber-optic or pipeline connections. 

We evaluate the performance of an incremental solution~$\pi$ in terms of the worst-case approximation guarantee it provides compared to the optimal budgeted Steiner-tree for all budgets $B \in \R_{\geq 0}$.
Specifically, for a given budget $B > 0$ and an incremental solution $\smash{\pi = (e_{\pi(1)},\dots,e_{\pi(l)})}$ computed by some algorithm~$\ALG$, 
let the rooted subtree induced by the set of edges $\smash{\bigl\{e_{\pi(i)} : i \in \{1,\dots,l\}, \sum_{j=1}^{i} c(e_{\pi(j)}) \leq B\bigr\}}$ be~$\ALG(B)$.
We aim to maximize the profit $p(\ALG(B))$ across all budgets $B \in \R_{\geq 0}$.

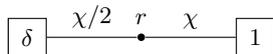
\begin{figure}[b]
\centering
\begin{tikzpicture}\small
         \tikzset{mininode/.style = {circle, fill, minimum size=2pt, inner sep=1pt},
            weightnode/.style = {rectangle, draw, minimum width=0.5cm,
                                minimum height = 0.5cm},
        }
            \useasboundingbox (-1.75,-1) rectangle (1.75,1.25);
            \node[mininode,label=above:{$r$}] at (0,0) (r) {};
            \node[weightnode] at (-1.5,0) (A) {$\delta$};
            \node[weightnode] at (1.5,0) (B) {$1$};
    
            \draw (r) -- node[above]{$\chi/2$} (A);
            \draw (r) -- node[above]{$\chi$} (B);
\end{tikzpicture}
\caption{Instance of the prize-collecting Steiner-tree problem with $\delta \ll 1$ and no $(0,\mu)$-competitive incremental solution: Staying competitive for budget~$B = \chi / 2$ requires to build the left edge first, but then we are not competitive for the budget~$B = \chi$.\label{fig:no-mult-competitive}}
\end{figure}

Observe that, in general, no incremental solution can stay close to the prizes collected by~$\OPT(B)$ across all budgets simultaneously (see Figure~\ref{fig:no-mult-competitive}).
However, it turns out that it suffices to allow an \emph{additive} slack in the budget constraint to obtain an incremental solution with \emph{multiplicative} approximation guarantee.
Formally, we call an algorithm \emph{$(\alpha,\mu)$}-competitive for $\alpha \geq 0$ and $\mu \geq 1$ if the incremental solution produced (for any instance of the problem) guarantees that
$\mu\,p(\ALG(B + \alpha)) \geq p(\OPT(B)) \textrm{ for all } B \in \R_{\geq 0}$.

\subsection{Our Results}

We first consider the incremental prize-collecting Steiner-tree problem on trees and analyze the density-greedy algorithm introduced by Alpern and Lidbetter~\cite{AlpernL13} in a different context.
Roughly, the algorithm repeatedly considers subtrees of maximum density, adds them to the incremental solution and contracts the corresponding subgraphs.
We show that this algorithm constructs an incremental solution that lags behind the optimum solution by some additive slack in the budget, but otherwise collects the same prize.
Figure~\ref{fig:no-mult-competitive} shows that we need to allow a slack depending on the \emph{eccentricity} of the root vertex $\chi := \max\{\ell(v) : v \in V\}$ of $G$, where~$\ell(v)$ denotes the minimum cost among all $r$-$v$-paths in~$G$.
We show that the density-greedy algorithm requires slack exactly~$\chi$.

\begin{restatable}{theorem}{thmresultgreedytrees}
\label{thm:result_greedy_trees}
    On trees, the density-greedy algorithm can be implemented in polynomial time and is $(\alpha,\mu)$-competitive for any finite $\mu \geq 1$ if and only if $\alpha \geq \chi$.
\end{restatable}

In particular, the density-greedy algorithm is $(\chi,1)$-competitive on trees.
We show that this is best possible.

\begin{restatable}{theorem}{thmresultLBtrees}
\label{thm:result_LB_trees}
    There exists no $(\alpha,\mu)$-competitive algorithm for~$\alpha < \chi$ and $\mu < \chi/\alpha$ or for~$\alpha < \chi/2$ and~$\mu \geq 1$, even on trees.
\end{restatable}

We proceed to generalize the density-greedy algorithm in order to apply it to non-tree networks. Here, an additive slack in budget equal to the eccentricity of the root no longer suffices for it to stay close to the optimum solution.
We obtain, however, an approximation guarantee when granting an additive slack equal to the maximum cost~$\lrp$ of any vertex-disjoint $r$-$v$-path in~$G$. 

\begin{restatable}{theorem}{thmresultgreedygraphs}
\label{thm:result_greedy_graphs}
    The density-greedy algorithm is $(\lrp,2)$-competitive, but not $(\chi,\mu)$-competitive for any~$\mu \geq 1$.
\end{restatable}

While the density-greedy algorithm cannot be implemented in polynomial time, unless $\mathsf{P}=\mathsf{NP}$, we give an approximate variant of the algorithm that runs in polynomial time at the expense of increasing the approximation factor from $2$ to $3$.

\begin{restatable}{corollary}{corresultgreedygraphs}
An approximate variant of the density-greedy algorithm runs in polynomial time and is~$(\gamma,3)$-competitive. 
\end{restatable}

We further prove that no algorithm can reach the optimum on general graphs if the additional slack in budget only depends on the length of the longest vertex-disjoint paths from the root.

\begin{restatable}{theorem}{thmlowerboundgraph}
For every $(\alpha,\mu)$-competitive algorithm with $\alpha$ only depending on $\lrp$, it holds that $\mu\geq 17/16$.
\end{restatable}

Note that since $\chi \leq \gamma$, this result also precludes the existence of an $(\alpha,\mu)$-competitive algorithm with $\alpha$ depending only on $\chi$ and $\gamma$ and $\mu < 17/16$. 

Finally, we design and analyze a capacity-scaling algorithm that is inspired by the incremental maximization algorithm of Bernstein et al.~\cite{BernsteinD0H22} and combines this approach with the density-greedy algorithm.
Instead of relying on subtrees of maximum density, the algorithm uses optimum solutions for the budget-constrained prize-collecting Steiner-tree problem for exponentially growing budgets.
These solutions are concatenated to give a solution for the incremental prize-collecting Steiner-tree problem.
We obtain the following result.

\begin{restatable}{theorem}{mainthmscalingmoreheadstart}
The capacity-scaling algorithm is $\smash{\bigl((4\ell\!-\!1)\chi,\frac{2^{\ell+2}}{2^\ell-1}\bigr)}$-competitive for every~$\ell \in \N$.
\end{restatable}

In particular, our algorithm guarantees an $8$-approximation when allowed an additional~$3 \chi$ in budget and its approximation guarantee approaches~$4$ for increasing additive slack.

\subsection{Related Work}

Our work is based on the cardinality-constrained incremental maximization framework of Bernstein et al.~\cite{BernsteinD0H22}.
They considered monotone augmentable objective functions (a class of functions containing monotone submodular functions) subject to a growing cardinality constraint and devise a cardinality-scaling algorithm that is~$2.618$-competitive.
Disser et al.~\cite{DisserKSW23} expanded on this result by giving an improved lower bound of~$2.246$ and considering a randomized scaling approach for the problem that achieves a randomized competitive ratio of~$1.772$.
The framework was extended to a budget-constrained variant by Disser et al.~\cite{DisserKlimmLutzWeckbecker-23}.
A detailed treatment of incremental maximization was given by Weckbecker~\cite{Weckbecker23}.

The Steiner-tree problem is known to be $\mathsf{NP}$-complete since Karp~\cite{Karp72}. In fact, it is $\mathsf{NP}$-hard to approximate the problem within a factor of $96/95$, as shown by Chlebík and Chlebíková~\cite{ChlebikC08}. The current best approximation algorithm with an approximation guarantee of $\ln(4)+\epsilon$ is due to Byrka et al.~\cite{ByrkaGRS13} improving upon previous approximation algorithms (e.g., Hougardy and Prömel~\cite{HougardyP99}; Robins and Zelikovsky~\cite{RobinsZ05}). Based on earlier work of Balas~\cite{Balas89}, the prize-collecting variant of the problem was first studied by Bienstock et al.~\cite{BienstockGSW93} who gave a 3-approximation for the objective to minimize the sum of the edge costs and the weights of the non-collected vertices. This was later improved to a $(2-\frac{1}{n})$-approximation by Goemans and Williamson~\cite{GoemansW95}, where $n$ is the number of vertices of the graph, and to~$1.9672$ by Archer et al.~\cite{ArcherBHK11}.
In a recent paper by Ahmadi et al.~\cite{ahmadi2024prize} the approximation ratio was further reduced to $1.7994$.
Johnson et al.~\cite{JohnsonMP00} discussed further natural variants of the problem: net-worth, budget, and quota. The \emph{net-worth problem} where the task is to maximize the total collected prize minus the total cost is $\mathsf{NP}$-hard to approximate by any constant factor (Feigenbaum et al.~\cite{FeigenbaumPS01}).
In the \emph{budget problem}, the task is to find a tree that maximizes the total collected prize collected subject to a budget constraint on the total cost of the tree.
Constant-factor approximations are only known for the unrooted variant of the problem \cite{JohnsonMP00,Levin04,PaulFFSW20}; see also the discussion by Paul et al.~\cite{Paul23erraturm}. For a problem more general than the rooted version, Ghuge and Nagarajan~\cite{GhugeN22} gave a quasi-polynomial poly-logarithmic approximation.
In the \emph{quota problem}, the task is to find a minimum cost tree that collects at least a given total prize. 
The quota version where all vertices have a prize of $1$ is known as the $k$-MST problem. The best known approximation for the $k$-MST problem is a $2$-approximation due to Garg~\cite{Garg05} improving upon previous approximation algorithms (e.g., Arya and Ramesh~\cite{AryaR98}; Blum et al.~\cite{BlumRV99}; Garg~\cite{Garg96}). As discussed by Johnson et al.~\cite{JohnsonMP00}, algorithms for the $k$-MST problem that rely on the primal-dual algorithm of Goemans and Williamson~\cite{GoemansW95} (such as the one by Garg~\cite{Garg05}) extend to the general quota problem so that the best known approximation factor for the quota problem is $2$ as well; see also the discussion by Griesbach et al.~\cite{GriesbachHKS23}.
An incremental version of the $k$-MST problem  has been studied by Lin et al.~\cite{LinNRW10}. Here, an incremental solution is a sequence of edges such that the cost of the smallest prefix of edges that connects at least $k$ vertices approximates the value of the optimal $k$-MST solution for all $k \in \mathbb{N}$. They noted that a deterministic $8$-competitive and a randomized $2e$-competitive incremental solution are implicit in the works on the minimum latency problem \cite{BlumCCPRS94,GoemansK98} and  gave improved incremental solutions with competitive ratios $4$ and $e$ for the deterministic and randomized case, respectively. 

The more general Steiner forest problem asks for given pairs of vertices to be connected at minimal cost. This problem admits a $(2-\frac{1}{n})$-approximation (Agrawal et al.~\cite{AgrawalKR95}; Goemans and Williamson~\cite{GoemansW95}).
For a prize-collecting version of the problem, Ahmadi et al.~\cite{AhmadiGHJM24} recently gave a $2$-approximation improving the $2.54$-approximation by Hajiaghayi and Jain~\cite{HajiaghayiJ06} whose approach has been generalized by Sharma et al.~\cite{SharmaSW07} towards more general submodular penalty functions.

Further related are dynamic and universal variants of the Steiner-tree problem.
In the dynamic Steiner-tree problem, vertices are added or deleted in an online fashion and a network spanning the active vertices has to be maintained with approximately optimal costs at all times. If only vertex additions are allowed, the natural greedy algorithm yields a $\mathcal{O}(\log n)$-approximation under $n$ additions as shown by Imase and Waxman~\cite{ImaseW91}. Gu et al.~\cite{GuG016} gave an algorithm that adds an edge and swaps an edge per addition and maintains a constant approximation. Naor et al.~\cite{NaorPS11} studied a further variant of the problem where also vertices have costs.
Xu and Moseley~\cite{0002M22} gave a learning-augmented algorithm that uses a prediction about which terminals will be added eventually.
When only deletions are allowed, Gupta and Kumar~\cite{GuptaK14} showed how to maintain a constant approximation while changing a constant number of edges per deletion.
Jia et al.~\cite{JiaLNRS05} introduced the notion of universal Steiner-trees which are rooted spanning trees containing an approximately minimal Steiner-tree for any subset of vertices. Universal Steiner-trees with poly-logarithmic approximation factor have been devised by Busch et al.~\cite{BuschCFHHR23}.

\section{Incremental Prize-Collecting Steiner-Tree on Trees} \label{sec:greedy_trees}

In this section we introduce the \emph{density-greedy algorithm} and prove that it constructs a~$(\chi,1)$-competitive algorithm for the incremental prize-collecting Steiner-tree problem on trees. We further prove that this is the best possible approximation guarantee, i.e., there is no~$(\alpha,1)$-competitive algorithm for trees where~$\alpha<\chi$.
Throughout this section, we assume that~$G=(V,E)$ is a tree with root vertex $r$ and prize~$p(r)=0$.

The density-greedy algorithm uses the concept of a \emph{min-max subtree} introduced by Alpern and Lidbetter~\cite{AlpernL13} in the context of the expanding search problem. For a non-empty rooted subtree $T \in \T$ of~$G$, its density $\dens(T)$ is defined as the ratio of the prizes that are collected by $T$ and the cost of $T$, i.e., $\dens_G(T) \coloneqq p(T) / c(T)$.
When the graph $G$ is clear from context, we may drop the index $G$.
For the empty rooted subtree $T_{\emptyset} \coloneqq (\{r\},\emptyset)$, we set~$\dens(T_{\emptyset}) \coloneqq 0$.
We further say that $T \in \mathcal{T}$ has \emph{maximum density} if $\dens(T) \geq \dens(T')$ for all~$T' \in \T$.
Let
$\M \coloneqq \{T \in \T : \dens(T) \geq \dens(T') \text{ for all } T' \in \T\}$
be the set of all trees with maximum density $d^*$.
The inclusion-wise minimal elements of $\M$ are called the \emph{min-max-subtrees} of the graph.
The following properties of the set $\M$ are due to Alpern and Lidbetter~\cite{AlpernL13}.

\pagebreak[4]

\begin{lemma}[cf.~\cite{AlpernL13}, Lemma~2]
    \label{lem:alpern_and_lidbetter}
    The set $\M \cup (\{r\},\emptyset)$ is closed under union and intersection. Furthermore, the following holds:
    \begin{enumerate}[(i)]
        \item \label{item:alpern1} Distinct min-max subtrees are (edge-)disjoint.
        \item \label{item:alpern2} Every branch at $r$ of $T \in \M$ has density $\dens^*$.
        \item \label{item:alpern3} Every min-max subtree $T$ of $G$ has only one branch at $r$.
    \end{enumerate}
\end{lemma}

We are now in position to explain the basic idea of the \emph{density-greedy algorithm} (cf.~\Cref{alg:dg-tree}). 
At the beginning, we start with the empty solution $\pi= ()$.
As long as the tree spanned by the edges in~$\pi$ does not collect the total prize of $G$, all edges contained in the current solution are contracted and we compute a min-max subtree in the contracted subgraph. By \Cref{lem:alpern_and_lidbetter}, this min-max subtree has a unique edge that is adjacent to the root. This edge is added to~$\pi$ and the whole procedure is repeated until $\pi$ collects the total prize of~$G$.
The density of the root vertex is defined to be zero. Hence, \Cref{alg:dg-tree} always terminates and returns a solution $\pi$ that spans all vertices in $V^*$.
In the remaining part of this section, we denote by $\pi$ the incremental solution returned by \Cref{alg:dg-tree}.
Recall that, for some cost budget $B\in \R_{\geq 0}$, we denote by $\ALG(B)$ the tree that is induced by the maximum prefix of $\pi$ that still obeys the cost budget $B$.

\begin{algorithm}[tb]
\caption{The density-greedy algorithm for trees}\label{alg:dg-tree}
\begin{algorithmic}[1]
\State $\pi \gets ()$
\State $\overbar{T} \gets (\{r\}, \emptyset)$
\While{$p_{\contract{G}{\overbar{T}}}(\contract{G}{\overbar{T}}) \neq 0$}
    \State fix arbitrary min-max subtree $T^*$
    of $\contract{G}{\overbar{T}}$ \label{it:compute-tree}
    \State denote by $e^*$
    the unique edge in $T^*$
    incident to $r$ \label{it:choose-e}
    \State denote by $e\in E$ the edge in $G$ that corresponds to $e^*$ in $\contract{G}{\overbar{T}}$
    \State append $e$
    to $\pi$ \label{it:append-e}
    \State $\overbar{T} \gets \overbar{T}\cup \{e\}$
\EndWhile
\State \Return $\pi$
\end{algorithmic}
\end{algorithm}

The algorithm maintains a graph where the edges of a rooted tree are contracted.
For a formal analysis of the algorithm, we introduce the following notion regarding contractions. 
 For some $T\in \T$, we obtain the \emph{contracted graph} $\contract{G}{T} = (V_{\contract{G}{T}},\ E_{\contract{G}{T}})$ by contracting all edges of $T$ into the root. 
 The graph $\contract{G}{T}$ then consists of $|V|-|E_T|$ vertices including the root vertex and $|E|-|E_T|$ edges, possibly including parallel edges and self-loops.
The vertex prizes in $\contract{G}{T}$ remain the same as the ones in the original graph~$G$ where the prize of the root is set to $0$, i.e.,~$p_{\contract{G}{T}}(v) \coloneqq p(v)$ for all $v\in V_{\contract{G}{T}}\setminus\{r\}$ and~$p_{\contract{G}{T}}(r) \coloneqq 0$.
Also any edge that is not contracted keeps its original cost, i.e., $c_{\contract{G}{T}}(e)\coloneqq c(e)$ for all~$e\in E_{\contract{G}{T}}$.
By keeping parallel edges and self-loops, we obtain the property that each edge in $G$ is in a one-to-one correspondence to either an edge in $T$ or to an edge in the contracted graph~$G/T$. This property will be of particular use for the following definition.
Let $S=(V_S,E_S)$ be a subgraph of $G$.
For a rooted subtree $T\in\T$ of $G$, we denote by $E_{\contract{S}{T}}$ the edges of~$S$ that have a one-to-one correspondence to an edge in the contracted graph $\contract{G}{T}$, i.e., $E_{\contract{S}{T}}=E_S\setminus E_T$. Then the \emph{contracted subgraph} $\contract{S}{T}$ is the subgraph of $\contract{G}{T}$ that is induced by the edge set $E_{\contract{S}{T}}\subseteq E_{\contract{G}{T}}$.
In particular, we obtain $p_{\contract{G}{T}}(\contract{S}{T})=p(S)-p(V_S \cap V_T)$ and 
$c_{\contract{G}{T}}(\contract{S}{T})= c(S)-c(E_S \cap E_T)$.

\begin{figure}
    \centering
    \begin{subfigure}{0.48\textwidth}
        \centering
        \begin{tikzpicture}
         \tikzset{mininode/.style = {circle, fill, minimum size=2pt, inner sep=1pt},
            weightnode/.style = {rectangle, draw, minimum width=0.5cm,
                                minimum height = 0.5cm},
        }
            \node[mininode, label=below:$r$] at (0,0) (r) {};
            \node[mininode] at (-1,0) (v1) {};
            \node[mininode] at (1,0) (v2) {};
            \node[mininode] at (0,1) (v7) {};
            \node[mininode] at (-2,0) (v9) {};
            \node[mininode] at (2,0) (v3) {};
            \node[mininode] at (3,0) (v4) {};
            \node[mininode] at (2,-1) (v5) {};
            \node[mininode] at (1,1) (v6) {};
            \node[mininode] at (1,2) (v8) {};
            \node[mininode] at (-3,-1) (v10) {};
    
            \draw (r) -- (v1);
            \draw[Red] (r) -- (v2);
            \draw (r) -- (v7);
            \draw (v1) -- (v9);
            \draw[rounded corners, MidnightBlue] (1.2,0.2) -- (0.2,0.2) -- (0.2,1.2) -- (-0.2,1.2) -- (-0.2,0.2) -- (-2.2,0.2) -- (-2.2,-0.4) --node[below] {$T$} (1.2,-0.4) -- cycle;
    
            \draw[ForestGreen] (v9) -- (v10);
            \draw[ForestGreen] (v2) -- (v3);
            \draw[ForestGreen] (v3) --node[below] {$C^+$} (v4);
            \draw[ForestGreen] (v3) -- (v5);
    
            \draw[Red] (v2) -- (v6);
            \draw[Red] (v3) -- (v6);
            \draw[Red] (v6) -- (v7);
            \draw[Red] (v6) --node[right] {$S$} (v8);
            \draw[Red] (v7) -- (v9);
        \end{tikzpicture}
    \end{subfigure}
    \begin{subfigure}{0.48\textwidth}
        \centering
        \begin{tikzpicture}
         \tikzset{mininode/.style = {circle, fill, minimum size=2pt, inner sep=1pt},
            weightnode/.style = {rectangle, draw, minimum width=0.5cm,
                                minimum height = 0.5cm},
        }
            \node[mininode, label=below:$r$] at (0,0) (r) {};
            \node[mininode] at (1,0) (v3) {};
            \node[mininode] at (2,0) (v4) {};
            \node[mininode] at (1,-1) (v5) {};
            \node[mininode] at (0.5,1) (v6) {};
            \node[mininode] at (0.5,2) (v8) {};
            \node[mininode] at (-1,-1) (v10) {};
    
            \draw[ForestGreen] (r) -- (v10);
            \draw[ForestGreen] (r) -- (v3);
            \draw[ForestGreen] (v3) --node[below] {$C$} (v4);
            \draw[ForestGreen] (v3) -- (v5);
    
            \draw[Red] (r) to[bend right] (v6);
            \draw[Red] (r) to[bend left] (v6);
            \draw[Red] (v3) -- (v6);
            \draw[Red] (v6) --node[right] {$\contract{S}{T}$} (v8);
            \draw[Red] (r) to [min distance=2cm, looseness=30, in=210, out=150, loop]  ();
        \end{tikzpicture}
    \end{subfigure}
    \caption{The left figure shows graph $G$ with tree $T\in \T$ and the right figure shows the contracted graph $\contract{G}{T}$. The contracted subgraph $\contract{S}{T}$ contains parallel edges and loops. The extension $C^+$ of the connected subgraph $C\subset \contract{G}{T}$ is not connected.
    \label{fig:contracted-graph}}
\end{figure}
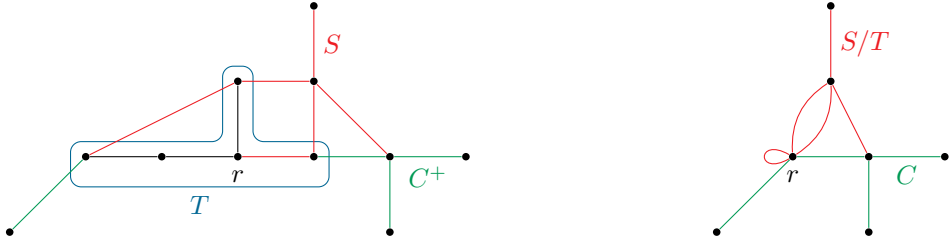

To reverse the contraction, we introduce the notion of an extended subgraph.
To this end, let $T\in \T$ and let $C=(V_C,E_C)$ be a subgraph of the contracted graph $\contract{G}{T}$.
We denote by $E_{C^+}$ the edges of $E$, that are in one-to-one correspondence to an edge in $E_C$. Note that~$|E_C|=|E_{C^+}|$.
Then, the \emph{extended subgraph} $C^+=(V_{C^+},E_{C^+})$ (or \emph{extension of~$C$}) is the subgraph of $G$ that is induced by the edge set $E_{C^+}$.
Observe, that $C^+$ is not necessarily connected in $G$, even if $C$ is connected in $\contract{G}{T}$. See \Cref{fig:contracted-graph} for an illustration.
However, if $C$ is connected and has at most one vertex adjacent to the root vertex, then $C^+$ is also connected.
We emphasise that in general, extending a contracted subgraph does not necessarily yield the original subgraph, i.e., for a subgraph $S$ of $G$ and some $T\in\T$ we have $(\contract{S}{T})^+\subseteq S$ and equality only holds when $E_S\cap E_T=\emptyset$.
However, repeating the contraction of $T$, then yields the same subgraph of $\contract{G}{T}$, i.e., $\contract{(\contract{S}{T})^+}{T}=\contract{S}{T}$.

\begin{figure}
    \centering
    \begin{tikzpicture}
     \tikzset{mininode/.style = {circle, fill, minimum size=2pt, inner sep=1pt},
        weightnode/.style = {rectangle, draw, minimum width=0.5cm,
                            minimum height = 0.5cm},
    }
        \node[mininode,label=above:{$r_T$}] at (-0.25,0.25) (r) {};
        \node[mininode,label=right:{\textcolor{MidnightBlue}{$v$}}, MidnightBlue] at (1,-1) (v) {};
        \node[mininode] at (0.2,-2) (v_1) {};
        \node[mininode] at (2,-2) (v_2) {};

        \draw (r) -- (-0.55,-0.05);
        \draw (r) -- (0.05,-0.05);
        \draw (0.5,-0.5) -- (v);
        \draw (v) -- (v_1);
        \draw (v) -- (0.8,-1.8);
        \draw[Red] (v) --node[above right] {$e$}  (v_2); 
        \draw (v_1) -- (0,-2.4);
        \draw (v_1) -- (0.4,-2.4);
        \draw (v_2) -- (1.5,-2.5);
        \draw (v_2) -- (2,-2.5);
        \draw (v_2) -- (2.5,-2.5);

        \node[label=below:{\dots}] at (-0.55,-0.05) {}; 
        \node[label=below:{\dots}] at (0.3,-0.05) {}; 
        \node[label=below:{\dots}] at (0.8,-1.8) {}; 
        \node[label=below:{\dots}] at (0.2,-2.4) {}; 
        \node[label=below:{\dots}] at (2,-2.5) {}; 

        \draw[rounded corners, MidnightBlue] (0.4,-0.7) -- (1.6,-0.7) -- (4.1,-3.2) -- (-1,-3.2) --node[above left] {$T_v$} cycle;
        \draw[rounded corners, Red] (0.4,-0.8) -- (1.5,-0.8) --node[above right] {$T_e$} (3.7,-3) -- (0.8,-3) -- (1.6,-2) -- cycle;        
    \end{tikzpicture}
    \caption{For the tree $T$ with root $r_T$, the branch $T_e$ rooted in edge $e$ is shown in red and the branch $T_v$ rooted in vertex $v$ is shown in blue.}
    \label{fig:rooted_subtree}
\end{figure}
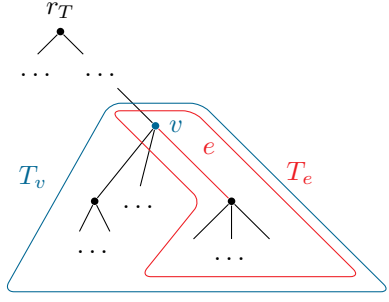

For the analysis of the density-greedy algorithm, we need to generalize the notion of density to subtrees that are not necessarily rooted. To this end, let $T=(V_T,E_T)$ be a subtree of $G$ and let $r_T\in V_T$ be the unique vertex of~$T$ that has the shortest distance to $r$ in $G$.
Then, the \emph{density of $T$ over $G$} is defined as $\dens_{G}(T)\coloneqq \bigl(p(T) - p(r_T)\bigr)\big/ c(T)$.
Similarly, for a collection of disjoint subtrees $F = \{T_1,\dots,T_k\}$, we set 
$ \dens_{G}(F) \coloneqq \bigl(\sum_{i=1}^k p(T_i) - p(r_{T_i})\bigr) \big/ \bigl( \sum_{i=1}^k c(T_i) \bigr)$ where again, for $i \in \{1,\dots,k\}$, the vertex $r_{T_i}$ is the unique vertex in $T_i$ with the smallest distance to $r$ in $G$. 
In order to avoid division by $0$, we set the density of a tree or forest to $0$ whenever the edge set is empty.

Let $T=(V_T,E_T)$ be a subtree of $G$ with root vertex $r_T$.
For an edge $e\in E_T$ we define the \emph{branch}~$T_e$ rooted in edge $e$ as the subtree of $T$ that is induced by the endpoints of $e$ and all vertices $u\in V_T$ such that~$e$ lies on the unique $r_T$-$u$-path in $T$, see \Cref{fig:rooted_subtree} for an illustration.
Analogously, for a vertex $v\in V_T$ we define the \emph{branch} $T_v$ rooted at $v$ as the subtree of $T$ that is induced by vertex $v$ and all vertices $u\in V_T$ such that $v$ lies on the unique~$r_T$-$u$-path in $T$, see \Cref{fig:rooted_subtree} for an illustration.

We proceed to show that for a min-max subtree, every branch $T_e$ rooted at some edge~$e$ has at least the same density as $T$.

\begin{restatable}{lemma}{lemhigherdensityininducedsubtrees}
    \label{lem:higher_density_in_induced_subtrees}
    Let $T$ be a min-max subtree of $G$. Then, for every edge $e$ in $T$ we have for the branch $T_e$ of~$T$ rooted at $e$ that $\dens_G(T_e) \geq \dens(T)$.
    Strict inequality holds, whenever $T\neq T_e$.
\end{restatable}

\begin{proof}
    Let $T=(V_T,E_T)$ and $T_e=(V_{T_e},E_{T_e})$. 
    By \Cref{lem:alpern_and_lidbetter}, $T$ has exactly one edge that is adjacent to $r$. Hence, if $r\in e$, we obtain $T=T_e$ and the statement is trivial. 
    Otherwise, consider the subgraph $F=(V_F,E_F)$ of $T$ that is induced by the edges $E_F \coloneqq E_T\setminus E_{T_e}$. Then $\emptyset\neq F\subset T$ and $r\in V_F$.
    Since $T$ is a min-max subtree of $G$, we have $\dens(F)<\dens(T)$.
    With $p(r)=0$, we obtain
    \begin{align*}
        \dens_G(T_e)=\frac{p(T)-p(F)}{c(T)-c(F)}
        = \frac{\dens(T)c(T)-\dens(F)c(F)}{c(T)-c(F)}
        >\frac{\dens(T)\bigl(c(T)-c(F)\bigr)}{c(T)-c(F)}
        =\dens(T),
    \end{align*}
    concluding the proof.
\end{proof}

The next lemma shows an important property about the order in which the density-greedy algorithm adds edges to the solution. In particular, it implies that all edges of the extension~$T^+$ of a min-max tree computed by the density-greedy algorithm are added to the solution~$\pi$ consecutively.

\begin{restatable}{lemma}{lemcompleteminmaxtree}
    \label{lem:complete_min-max_tree}
    Let $T^*$ be a min-max subtree of $\contract{G}{\overbar{T}}$ for some $\overbar{T}\in \T$ as computed in an iteration of the density-greedy algorithm for trees.
    Then, until all edges of the extension $T^+$ of $T^*$ are added to $\pi$, the extended subgraphs of computed min-max subtrees are subtrees of $T^+$ and have a larger density than $T^+$.
\end{restatable}
\begin{proof}
    Let $\overbar{T},T^*$, and $T^+=(V_{T^+},E_{T^+})$ be as stated in the lemma.
    We set $\ell\coloneqq |E_{T^+}|$ and denote by~$e_1,\dots,e_\ell$ the next $\ell$ edges as they are added to the solution $\pi$ after $T^*$ was computed in Step \ref{it:compute-tree}.
    For each $i\in \{1,\dots,\ell\}$ we denote by $T_i^*$ the min-max tree that was computed and fixed in Step \ref{it:compute-tree} right before $e_i$ was added to $\pi$, and denote by $T_i^+$ the extension of $T_i^*$.
    In particular, this yields $T^*=T_1^*$ and $T^+=T_1^+$. To prove the statement, we need to show that $T_i^+\subseteq T^+$ and $\dens_G(T^+)<\dens_G(T_i^+)$ for all~$i\in \{2,\dots,\ell\}$.
    
    First, we prove the following claim: If for some fixed $i\in\{2,\dots,\ell\}$ we have $T_j^+\subseteq T^+$ for all $1\leq j<i$, then 
    \begin{align} \label{eq:increasing-density-min-max-subtrees}
        \dens_G(T^+)<\dens_G(T_i^+).
    \end{align}
    To this end, let $i\in\{2,\dots,\ell\}$ be such that $T_j^+\subseteq T^+$ for all $1\leq j<i$.
    We define the subgraph~$S_i\subseteq G$ to be the subgraph induced by the edges $\{e_1,\dots,e_{i-1}\}$.
    Since $E_{S_i}\cap E_{\overbar{T}}=\emptyset$, all edges of~$S_i$ have a corresponding edge in the contracted graph $\contract{G}{\overbar{T}}$.
    Thus, we obtain $\dens_{\contract{G}{\overbar{T}}}(\contract{S_i}{\overbar{T}})=\dens_G(S_i)$.
    Furthermore, the assumption that $T_j^+\subseteq T^+$ for all $1\leq j<i$ implies that the edges $\{e_1,\dots,e_{i-1}\}$ are contained in $T^+$ and, therefore, $\contract{S_i}{\overbar{T}}\subsetneq T^*$.
    Since $T^*$ is a min-max subtree of $\contract{G}{\overbar{T}}$, it follows that $\dens_{\contract{G}{\overbar{T}}}(\contract{S_i}{\overbar{T}})<\dens_{\contract{G}{\overbar{T}}}(T^*)$.
    Next, consider the forest $S_{-i}$ of $G$ that is induced by the edge set $E_{T^+}\setminus S_i$.
    Then, $\dens_{G}(S_i)<\dens_{G}(T^+)$ and~$T^+=S_i \cup S_{-i}$ implies $\dens_G(S_{-i})>\dens_G(T^+)$.
    As $S_{-i}$ contains no edges of $S_i$ or $\overbar{T}$, all edges of $S_{-i}$ have a corresponding edge in the contracted graph $\contract{G}{(S_i\cup \overbar{T})}$.
    With $T_i^*$ being a min-max subtree of~$\contract{G}{(S_i\cup \overbar{T})}$, we obtain
    \begin{align*}
        \dens_G(T^+)<\dens_G(S_{-i})=\dens_{\contract{G}{(S_i\cup \overbar{T})}}\bigl(\contract{S_{-i}}{(S_i\cup \overbar{T})}\bigr)\leq \dens_{\contract{G}{(S_i\cup \overbar{T})}}(T_i^*)=\dens_G(T_i^+).
    \end{align*}
    Here, the last equality follows from $G$ being a tree and the fact that the edge set of $T_i^+$ is disjoint from the edge sets of $\overbar{T}$ and $S_i$.
    This concludes the proof of $(\ref{eq:increasing-density-min-max-subtrees})$.

    It remains to show that $T_i^+\subseteq T^+$ for all $i\in\{1,\dots,\ell\}$.
    Assume for contradiction that there exists some~$T_i^+=(V_{T_i^+},E_{T_i^+})$ with $T_i^+\nsubseteq T^+$. Out of these trees, let $T_i^+$ be the one with the smallest index.
    Then $T_j^+\subseteq T^+$ for all $j<i$ and together with Claim \eqref{eq:increasing-density-min-max-subtrees} we obtain $\dens(T^+)<\dens(T_i^+)$.
    Observe that~$T_i^+\nsubseteq T^+$ implies that there exists an edge $e\in E_{T_i^+}\setminus E_{T^+}$ that is incident to a vertex of $T^+$.
    We write $T_{i,e}^+$ for the branch of $T_i^+$ rooted in $e$.
    Since $T^+$ is the extension of the min-max subtree $T^*$, the density of $T_{i,e}^+$ over $G$ cannot be larger than the density of $T^+$ over $G$, i.e., $\dens_G(T_{i,e}^+)\leq \dens_G(T^+)$, as otherwise adding the corresponding edges of $T_{i,e}^+$ to $T^*$ would give a subtree of $\contract{G}{\overbar{T}}$ with strictly larger density than $T^*$. This would, however, contradict to $T^*$ being a min-max subtree of $\contract{G}{\overbar{T}}$.
    By \Cref{lem:higher_density_in_induced_subtrees}, we obtain
    $\dens_G(T_i^+)\leq \dens_G(T_{i,e}^+)\leq\dens_G(T^+)$.
    This gives the desired contradiction to $\dens_G(T^+)<\dens_G(T_i^+)$ and, thus, concludes the proof that $T_i^+\subseteq T^+$ for all $i\in\{1,\dots,\ell\}$.
    Together with Claim \eqref{eq:increasing-density-min-max-subtrees}, the lemma follows.
\end{proof}

Next, we give the first lower bound for the total prize that is collected by the density-greedy algorithm for trees for some specific cost budgets.
To do so, we denote by $T_1$ the min-max subtree of $G$ that is computed and fixed in the first iteration of \Cref{alg:dg-tree}.
Recall that $\ALG(B)$ denotes the tree induced by the maximum prefix of the solution computed by the density-greedy algorithm obeying the cost budget~$B$.

\begin{restatable}{lemma}{lemoptimalonfirsttree}
    \label{lem:optimal_on_first_tree}
    For every cost budget $B\leq c(T_1)$, we have
    $p\bigl(\ALG(B+\chi)\bigr)\geq \dens(T_1)B$.
    Furthermore, if $B=c(T_1)$, we have
            $p\bigl(\ALG(B)\bigr)=\dens(T_1)B$.
\end{restatable}
\begin{proof}
    By \Cref{lem:complete_min-max_tree}, the first $|T_1|$ edges of $\pi$ are exactly the edges of $T_1$. Since $T_1$ is a min-max subtree of $G$, we have $\dens(T_1)=\dens^*$ and thus, we obtain $p\bigl(\ALG(c(T_1)\bigr)=p(T_1)=\dens(T_1)c(T_1)=\dens^*c(T_1)$. This proves the second statement of the lemma and also the first statement for all cost budgets $B$ with $c(T_1)-\chi\leq B\leq c(T_1)$.

    It remains to consider the case where $B<c(T_1)-\chi$. Then $\ALG(B+\chi)$ is a proper subtree of $T_1$.
    We denote by $e^*$ the first edge in $\pi$ that exceeds the cost budget $B+\chi$ and is, thus, not contained in~$\ALG(B+\chi)$.
    Let $P_{e^*}$ be the unique path in $G$ that starts in the root vertex $r$ and ends with $e^*$.
    First, assume the edge set of $\ALG(B+\chi)\setminus P_{e^*}$ is empty. Then, $\ALG(B+\chi)\cup {e^*}$ is a simple path in $G$ that starts in $r$.
    Hence, its length can be no longer than $\chi$ which yields $B +\chi <\chi$; a contradiction to $B \geq 0$.

    Thus, we may assume that the edge set of the forest $F=(V_F,E_F):=\ALG(B+\chi)\setminus P_{e^*}$ is not empty.
    We want to show that
    \begin{align}\label{eq:subtrees-density}
        \dens_G(F) \geq \dens(T_1).
    \end{align}
    To that end, let $e\in E_F$.
    We denote by $T^*_e$ the min-max subtree that was fixed in the iteration of the density-greedy algorithm in which $e$ was added to the solution $\pi$ and denote by $T_e^+=(V^+_e,E^+_e)$ its extension.
    With \Cref{lem:complete_min-max_tree}, all edges of $T^+_e$ are added to $\pi$ consecutively.
    Since $e$ does not lie on the path~$P_{e^*}$ and $T^*_e$ has only one branch at the root, $e^*\notin E^+_e$.
    Together with $e^*$ being the first edge in $\pi$ that exceeds the cost budget $B+\chi$, it follows that all edges of $T^+_e$ are contained in $\ALG(B+\chi)$.
    Furthermore, \Cref{lem:complete_min-max_tree} implies that $\dens_G(T^+_e)>\dens_G(T_1)$ for all $e\in E_F$.
    Also, by \Cref{lem:complete_min-max_tree} we have either $T^+_e$ and $T^+_{e'}$ are edge-disjoint or one is a subtree of the other for $e, e' \in E_F$. So, the inclusion-wise maximal trees in~$\{ T^+_e : e \in F\}$ are edge-disjoint and partition the edges of $F$.
    Therefore, we obtain Claim \eqref{eq:subtrees-density} using that $\dens_G(T^+_e)>\dens_G(T_1)$ for each of these trees.
     
    We conclude that
    \begin{multline*}
        p\bigl(\ALG(B+\chi)\bigr)
        \geq p(F)
        =\dens_G(F)c(F)
        \geq \dens(T_1) c(F) \\
        =\dens(T_1)\bigl[c\bigl(\ALG(B+\chi)\bigr)+c(e^*)-c(P_{e^*})\bigr]
        >\dens(T_1)B,
    \end{multline*}
    where we used $c\bigl(\ALG(B+\chi)\bigr) +c(e^*) > B+\chi$ and $c(P_{e^*}) \leq \chi$ for the last inequality.
\end{proof}

\Cref{lem:complete_min-max_tree} states that the first $|T_1|$ edges of $\pi$ are the edges of $T_1$.
Let~$T_2^*$ be the $(|T_1|+1)$-st min-max subtree that is computed by the density-greedy algorithm. In particular, $T_2^*$ is the first min-max subtree whose extension $T_2$ contains no edge of $T_1$.
This inductively defines all pairwise disjoint min-max subtrees $T_1^*,\dots,T_k^*$ and their extensions $T_1,\dots,T_k$ such that~$V^*\subseteq \bigcup_{i=1}^kV_{T_i}$.
We define $\overbar{T}_i \coloneqq \bigcup_{j=1}^i T_j$ for $1 \leq i \leq k $ and set $\overbar{T}_0 \coloneqq (\{r\},\emptyset)$.
Note that, since $T_i^*$ is a min-max subtree of $\contract{G}{\overbar{T}_{i-1}}$, we have for the densities over $G$ that~$\dens_G(T_i \cup T_{i+1})
= \dens_{\contract{G}{\overbar{T}_{i-1}}} \bigl((T_i^* \cup T_{i+1}^*)\bigr)\leq 
\dens_{\contract{G}{\overbar{T}_{i-1}}} (T_i^*) = 
\dens_G(T_i)$.
In particular, this yields $\dens_G(T_{i+1}) \leq \dens_G(T_i)$ for all $0<i<k$.
By applying Lemma \ref{lem:optimal_on_first_tree} iteratively on $\contract{G}{\overbar{T}_{i}}$ for all~$i\in\{0,\dots,k-2\}$, we obtain the following Corollary.

\begin{corollary}
    \label{cor:boundalg}
    Let $B\in \R_{\geq 0}$ be a cost budget such that $\sum_{i=1}^{j-1} c(T_i) \leq B \leq \sum_{i=1}^{j} c(T_i)$ for some $1\leq j\leq k$.
    Then, we have
    \begin{align*}
        p\bigl(\ALG(B+ \chi)\bigr) \geq \sum_{i=1}^{j-1} \dens(T_i) c(T_i) + \dens(T_j) \left( B- \sum_{i=1}^{j-1} c(T_i) \right).
    \end{align*}
\end{corollary}

To give an upper bound on the prize collected by the optimum solution, we show the following lemma.

\begin{restatable}{lemma}{lemboundopt}
    \label{lem:boundopt}
    Let $T\in \T$ and $j\in \{1,\dots,k\}$. Then, we have
    \begin{align*}
        p(T) \leq \sum_{i=1}^{j-1} \dens(T_i) c(T_i) + \dens(T_j) \left( c(T)-\sum_{i=1}^{j-1} c(T_i) \right).
    \end{align*}
\end{restatable}
\begin{proof}
    By definition, $T_i^*=\contract{T_i}{\overbar{T}_{i-1}}$ is a min-max subtree of $\contract{G}{\overbar{T}_{i-1}}$ for all $i\in\{1,\dots,j\}$.
    Thus, we obtain the two lower bounds 
    \begin{align*}
        \dens_G(T_i)&=\dens_{\contract{G}{\overbar{T}_{i-1}}}(\contract{T_i}{\overbar{T}_{i-1}})\geq \dens_{\contract{G}{\overbar{T}_{i-1}}}\bigl(\contract{(T\cap T_i)}{\overbar{T}_{i-1}}\bigr) \quad \text{ and} \\
        \dens_G(T_i)&=\dens_{\contract{G}{\overbar{T}_{i-1}}}(\contract{T_i}{\overbar{T}_{i-1}})\geq \dens_{\contract{G}{\overbar{T}_{i-1}}} (\contract{T}{\overbar{T}_{i-1}})
    \end{align*}
    for all $i\in\{1,\dots,j\}$. Combining these two bounds, we obtain
    \begin{align}
        p(T) &= \sum_{i=1}^{j-1} p_{\contract{G}{\overbar{T}_{i-1}}}\bigl(\contract{(T \cap T_i)}{\overbar{T}_{i-1}}\bigr) + p_{\contract{G}{\overbar{T}_{j-1}}} \left(\contract{T}{\overbar{T}_{j-1}} \right) \notag \\
        &\leq \sum_{i=1}^{j-1} \dens(T_i) c_{\contract{G}{\overbar{T}_{i-1}}}\bigl(\contract{(T \cap T_i)}{\overbar{T}_{i-1}}\bigr) + \dens(T_j) c_{\contract{G}{\overbar{T}_{j-1}}}(\contract{T}{\overbar{T}_{j-1}}) \notag\\
        &= \sum_{i=1}^{j-1} \dens(T_i) c(T \cap T_i) + \dens(T_j) c\left(T \setminus \overbar{T}_{j-1}\right) \notag\\
        &= \sum_{i=1}^{j-1} \dens(T_i) c(T_i) + \sum_{i=1}^{j-1} \dens(T_i) \bigl[c(T \cap T_i)- c(T_i) \bigr] + \dens(T_j) c\bigl(T \setminus \overbar{T}_{j-1}\bigr) \notag\\
        &\leq \sum_{i=1}^{j-1} \dens(T_i) c(T_i) + \dens(T_j) \left[ \sum_{i=1}^{j-1} c(T \cap T_i) + c\bigl(T \setminus \overbar{T}_{j-1}\bigr) - \sum_{i=1}^{j-1} c(T_i) \right] \label{eq:boundopt}\\
        &= \sum_{i=1}^{j-1} \dens(T_i) c(T_i) + \dens(T_j) \left( c(T)-\sum_{i=1}^{j-1} c(T_i) \right) \notag
    \end{align}
    where for \eqref{eq:boundopt} we use $c(T \cap T_i)-c(T_i)\leq 0$ and the fact that $\dens(T_j)\leq \dens(T_i)$ for all $i\in\{1,...,j-1\}$.
\end{proof}

From Corollary \ref{cor:boundalg} and Lemma \ref{lem:boundopt} we obtain that the density-greedy algorithm is $(\chi,1)$-competitive. Together with a suitable lower bound, we obtain the following result.
\pagebreak[2]

\thmresultgreedytrees*
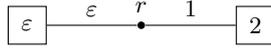
\begin{figure}
    \centering
        \centering
        \begin{tikzpicture}\small
         \tikzset{mininode/.style = {circle, fill, minimum size=2pt, inner sep=1pt},
            weightnode/.style = {rectangle, draw, minimum width=0.5cm,
                                minimum height = 0.5cm},
        }
            \useasboundingbox (-1.75,-1) rectangle (1.75,1);
            \node[mininode,label=above:{$r$}] at (0,0) (r) {};
            \node[weightnode] at (-1.5,0) (A) {$\varepsilon$};
            \node[weightnode] at (1.5,0) (B) {$2$};
    
            \draw (r) -- node[above]{$\varepsilon$} (A);
            \draw (r) -- node[above]{$1$} (B);
        \end{tikzpicture}
        
   \caption{\label{fig:lower-bound-trees-density} Lower bound used in the proof of \Cref{thm:result_greedy_trees}}
\end{figure}

\begin{proof}
Hermans et al.~\cite{HermansLM22} argue that maximum density subtrees  can be computed in polynomial time by the dynamic program given in \citep{AlpernL13} when the underlying graph is a tree. Hence, the density-greedy algorithm can be implemented in polynomial time on trees.

We first show that the density-greedy algorithm is $(\chi,1)$-competitive on trees, i.e., we show that $\smash{p\bigl(\ALG(B+\chi)\bigr) \geq p\bigl(\OPT(B)\bigr)}$ for every cost budget $B\in \R_{\geq0}$.
    If $B\geq \sum_{i=1}^{k} c(T_i)$, then we have $p\bigl(\ALG(B+\chi)\bigr) = p(G) \geq \OPT(B)$.
    Otherwise, let $j$ be maximal such that $\sum_{i=1}^{j-1} c(T_i) \leq c\bigl(\OPT(B)\bigr)\leq B$. Then, applying \Cref{lem:boundopt} and \Cref{cor:boundalg} yields
    \begin{align*}
        p\bigl(\OPT(B)\bigr) &\leq \sum_{i=1}^{j-1} \dens(T_i) c(T_i) + \dens(T_j) \left( c\bigl(\OPT(B)\bigr)-\sum_{i=1}^{j-1} c(T_i) \right)\\
        &\leq p\bigl(\ALG(c(\OPT(B))+ \chi)\bigr) \\
        &\leq p\bigl(\ALG(B+\chi)\bigr).
    \end{align*}
We proceed to show that the density-greedy is not $(\alpha,\mu)$-competitive for any $\mu \geq 1$ if $\alpha < \chi$. To this end, let $0\leq \alpha <1$.
    Consider the instance given in \Cref{fig:lower-bound-trees-density} for some $\varepsilon < 1-\alpha$. Then, $\chi=1$.
    The first min-max tree $T^*$ computed by the density-greedy algorithm only contains the edge with cost $1$ and has density~$\dens^*=2$.
    Thus, the prize collected by the density-greedy algorithm for any cost budget less than~$\chi=1$ is $0$, i.e, $p(\ALG(\varepsilon+\alpha)) = 0$ whereas $p(\OPT(\varepsilon))=\varepsilon$.
    Hence, we obtain $\mu \cdot p(\ALG(\varepsilon+\alpha))
        =0
        <\varepsilon
        =p(\OPT(\varepsilon))$ for all $\mu\geq 1$.
\end{proof}

We further prove that that there is no $(\alpha,1)$-competitive algorithm for trees for any~$\alpha<\chi$, which implies that the density-greedy algorithm gives the best possible approximation guarantee. In fact, we show the following stronger result.

\thmresultLBtrees*
\begin{figure}[tb]
\centering
    \begin{subfigure}[c]{0.45\textwidth}
        \centering
        \begin{tikzpicture}\small
         \tikzset{mininode/.style = {circle, fill, minimum size=2pt, inner sep=1pt},
            weightnode/.style = {rectangle, draw, minimum width=0.5cm,
                                minimum height = 0.5cm},
        }
            \useasboundingbox (-1.75,-1) rectangle (3,1);
            \node[mininode,label=above:{$r$}] at (0,0) (r) {};
    		\node[weightnode] at (-1.5,0) (A) {$2\chi$};
    		\node[weightnode] at (1.5,0.7) (B) {\!$\chi/k$\!};
    		\node[weightnode] at (1.5,-0.7) (C) {\!$\chi/k$\!};
    
            \draw (r) -- node[above]{$\chi$} (A);
            \draw (r) -- node[above]{$\chi/k$} (B);
            \draw (r) -- node[below]{$\chi/k$} (C);
            \node[label=right:{$k$ copies}] at (1.5,0.1) {\large{$\vdots$}};
    	\end{tikzpicture} 
        \caption{}
        \label{fig:lower-bound-tree-thm}
    \end{subfigure}
    \hfill
    \begin{subfigure}[c]{0.45\textwidth}
        \centering
        \begin{tikzpicture}\small
         \tikzset{mininode/.style = {circle, fill, minimum size=2pt, inner sep=1pt},
            weightnode/.style = {rectangle, draw, minimum width=0.5cm,
                                minimum height = 0.5cm},
        }
            \useasboundingbox (-1.75,-1) rectangle (1.75,1);
    		\node[mininode,label=above:{$r$}] at (0,0) (r) {};
    		\node[weightnode] at (-1.5,0) (A) {$1$};
    		\node[weightnode] at (1.5,0) (B) {$k$};
    
            \draw (r) -- node[above]{$\lfrac{1}{2}$} (A);
            \draw (r) -- node[above]{$1$} (B);
    	\end{tikzpicture} 
        \caption{}
        \label{fig:lower-bound-tree-prop}
    \end{subfigure}
    \caption{Lower bounds used in the proof of \Cref{thm:result_LB_trees}}
\end{figure}
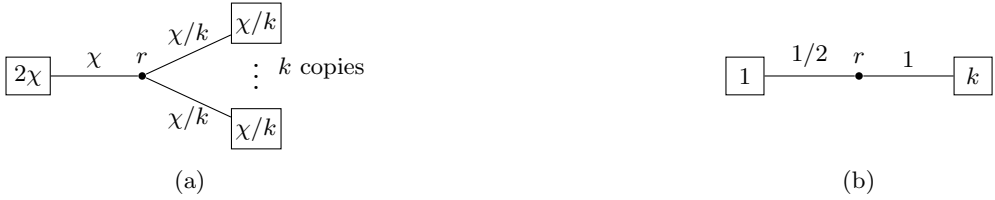

\begin{proof}
Let $\ALG$ be $(\alpha,\mu)$-competitive. 
We first show that when $\alpha \in [\chi/2,\chi)$, then we have $\mu < \chi/\alpha$.
To this end, consider the instance $G$ given in \Cref{fig:lower-bound-tree-thm} with eccentricity $\chi\in \R_{>0}$.
    Let $\pi=(e_1,\dots,e_{k+1})$ be the solution computed by $\ALG$ for the incremental prize-collecting Steiner-tree problem for some $\alpha\in\left[\chi/2,\chi\right)$ when applied on the instance~$G$.
    Let $e_{\ell+1}$ be the edge of cost $\chi$ for some $\ell\in\{0,\dots,k\}$.
    Then, $p(\{e_1,\dots,e_\ell\})=c(\{e_1,\dots,e_\ell\})=\chi\ell/k$.
    We proceed to distinguish two cases.
    
    First, we assume that $\chi\ell/k > \alpha$ for some $k \in \N$.
    In this case, the total cost of the first $\ell+1$ edges is bounded from below by $c(\{e_1,\dots,e_{\ell+1}\})=\chi+\frac{\ell}{k}\chi>\chi+\alpha$.
    For cost budget $B\coloneqq \chi$ we thus obtain that $p\bigl(\ALG(B+\alpha)\bigr)=\lfrac{\ell\chi}{k}\leq \chi$, while $p\bigl(\OPT(B)\bigr)=2\chi$.
    Hence, the multiplicative factor is at least~$\mu\geq \frac{p(\OPT(B))}{p(\ALG(B+\alpha))}\geq 2\geq \lfrac{\chi}{\alpha}$.

    For the second case, we assume that $\chi \ell /k \leq\alpha$ for every $k \in \N$.
    Consider budget $B\coloneqq \chi+\frac{\ell}{k}\chi-\delta-\alpha$ for some $0<\delta< \chi /k$.
    The prize collected by $\ALG$ is 
    $p\bigl(\ALG(B+\alpha)\bigr)=p\bigl(\ALG(\chi+\frac{\ell}{k}\chi-\delta)\bigr)= \chi\ell /k $, while 
    $p\bigl(\OPT(B)\bigr)=p\bigl(\OPT(\chi+\frac{\ell}{k}\chi-\delta-\alpha)\bigr)\geq \chi+\frac{\ell-2}{k}\chi-\alpha$.
    We rearrange these equations and use the fact that $\ALG$ is $(\alpha,\mu)$-competitive, $\lfrac{\alpha}{\chi}\geq \lfrac{1}{2}$, and $\ell\geq0$ to obtain
    \begin{align*}
        \ell=\frac{k}{\chi}\, p\bigl(\ALG(B+\alpha)\bigr)
        \geq \frac{k}{\chi\mu}\, p\bigl(\OPT(B)\bigr)
        \geq \frac{k}{\chi\mu}\left(\chi+\frac{\ell-2}{k}\chi-\alpha\right)
        \geq \frac{k}{\mu}\left(1-\frac{\alpha}{\chi}\right)-\frac{2}{\mu}.
    \end{align*}
    Thus, $k \rightarrow \infty$ implies $\ell \rightarrow \infty$.
    Finally, we obtain a lower bound on $\mu$ by
    \begin{multline*}
        \mu
        \geq \frac{p\bigl(\OPT(B)\bigr)}{p\bigl(\ALG(B+\alpha)\bigr)}
        \geq \left( \chi+\frac{\ell-2}{k}\chi-\alpha \right) \frac{k}{\ell \chi}
        =\frac{k}{\ell \chi}(\chi-\alpha)+\frac{\ell-2}{\ell}\\
        \geq \frac{\chi}{\alpha} -1+\frac{\ell-2}{\ell}
        \geq \frac{\chi}{\alpha}-\frac{2}{\ell},
    \end{multline*}
    where we used $\lfrac{\ell \chi}{k}\leq\alpha$ in the second to last inequality.
    Since $k \rightarrow \infty$ implies $\ell \rightarrow \infty$, we obtain $\mu \geq \lfrac{\chi}{\alpha}$ as claimed.

We proceed to show that if $\alpha < \chi/2$, then there is no finite $\mu \geq 1$ such that $\ALG$ is $(\alpha,\mu)$-competitive. 
To this end, consider the instance $G$ given in \Cref{fig:lower-bound-tree-prop} with $\chi=1$.
    Let $\pi$ be the solution obtained from an algorithm for the incremental prize-collecting Steiner-tree problem applied on $G$.
    Further, assume that $0\leq \alpha <\lfrac{1}{2}$.
    If the first edge in $\pi$ is the edge of cost $\lfrac{1}{2}$, we obtain
    $\mu\, p\bigl(\ALG(1+\alpha)\bigr)=\mu$, while $p\bigl(\OPT(1)\bigr)=k$.
    For a fixed $\mu$, we can set $k$ to some value strictly larger than $\mu$ and obtain $\mu\, p\bigl(\ALG(1+\alpha)\bigr)=\mu<k=p\bigl(\OPT(1)\bigr)$, i.e., the algorithm is not $(\alpha,\mu)$-competitive.
    However, if the first edge in $\pi$ corresponds to the edge of cost $1$, we obtain
    $p\bigl(\ALG(\lfrac{1}{2}+\alpha)\bigr)=0$, while $p\bigl(\OPT(\lfrac{1}{2})\bigr)=1$.
    Hence, the algorithm is not $(\alpha,\mu)$-competitive for any $\mu\geq 1$.
\end{proof}

\section{Incremental Prize-Collecting Steiner-Tree on General Graphs}\label{sec:dg-graphs}

In this section we study the incremental prize-collecting Steiner-tree problem on general graphs.
As a first step, we extend the density-greedy algorithm from the previous section such that we can apply it on general graphs, as well.
Afterwards we introduce a new algorithm based on capacity scaling and, finally, we give general lower bounds for the approximation factors of the incremental prize-collecting Steiner-tree problem.
Throughout this section, we assume that $G=(V,E)$ is a graph with root vertex $r$ and $p(r)=0$.

\subsection{Adapting the Density-Greedy Algorithm}

\begin{figure}
    \centering
    \begin{tikzpicture}\small
         \tikzset{mininode/.style = {circle, fill, minimum size=2pt, inner sep=1pt},
            weightnode/.style = {rectangle, draw, minimum width=0.5cm,
                                minimum height = 0.5cm},
        }
        \node[mininode,label=left:{$r$}] at (-1,0) (r) {};
        \node[mininode] at (0,0) (r') {};
        \node[mininode] at (2,0) (x) {};
        \node[mininode] at (1,-1) (u) {};
        \node[weightnode] at (3,0) (A) {$3$};
        \node[weightnode] at (1,-2) (S) {$2$};

        \draw[Red] (r) -- node[black, above]{$e$} 
        (r');
        \draw[MidnightBlue] (r') -- node[black, above]{$e_1$}
        (x);
        \draw[Red] (x) -- 
        (A);
        \draw[ForestGreen] (r') -- node[black, below left]{$e_2$} 
        (u);
        \draw[ForestGreen] (u) --
        (x);
        \draw[ForestGreen] (u) --
        (S);
    \end{tikzpicture}
    \caption{Assume all edges have unit cost. Then, $G$ has two min-max subtrees: $T_1$ consists of the blue and red edges, $T_2$ consists of the green and red edges. Both have maximum density $\dens^*=1$. Note that neither the union nor the intersection of $T_1$ and $T_2$ is a min-max subtree.\\
    If we slightly reduce the cost of edge $e_2$ to $\lfrac{3}{4}$, tree $T_2$ is the only min-max subtree. Hence, in the first iteration of \Cref{alg:dg-tree}, we contract edge $e$. However, in the second iteration, the tree spanned by the blue and red edges is the only min-max subtree and thus, \Cref{lem:complete_min-max_tree} does not hold for general graphs.
    \label{fig:intersecting-min-max-subtrees}}
\end{figure}
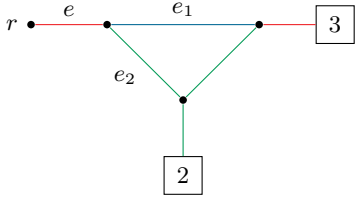

We start by generalizing statement \eqref{item:alpern2} and \eqref{item:alpern3} of \Cref{lem:alpern_and_lidbetter} to general graphs.
The proof is similar to the original one for trees but is given here for completeness.
\pagebreak[2]

\begin{restatable}{lemma}{alpernlidbetterforgeneralgraphs}\label{lem:alpern_and_lidbetter_gen}
    Let $T\in \M$ be a rooted subtree of $G$ with maximum density $\dens^*$. Then the following holds:
    \begin{enumerate}[(i)]
        \item \label{item:alperngraphs1} Every branch of $T$ at $r$ has density $\dens^*$.
        \item \label{item:alperngraphs2} If $T$ is a min-max subtree of $G$, $T$ has only one branch at $r$.
    \end{enumerate}
\end{restatable}
\begin{proof}
    For~(\ref{item:alperngraphs1}) we observe that if there is a branch at~$r$ of density smaller~$\dens^*$ then removing this branch from $T$ would yield a rooted tree with increased density contradicting that~$T$ has maximum density. For~(\ref{item:alperngraphs2}) assume that~$T$ has more than one branch at~$r$, then those branches are proper rooted subtrees of~$T$ and by~(\ref{item:alperngraphs1}) have density~$\dens^*$ contradicting the fact that $T$ has no proper rooted subtrees of maximum density.
\end{proof}

A significant difference to the setting on trees, however, is that on general graphs the min-max subtrees are not closed under intersection and union, i.e., statement \eqref{item:alpern1} of \Cref{lem:alpern_and_lidbetter} does not hold anymore, as illustrated in \Cref{fig:intersecting-min-max-subtrees}.
As a consequence we do not obtain \Cref{lem:complete_min-max_tree}, that is, the edges of a fixed min-max subtree are not necessarily added to the solution consecutively.
Since the analysis of the density-greedy algorithm crucially depends on this property, we need to make some adjustments to the algorithm when generalizing it in order to recover this property.
The \emph{density-greedy algorithm for general graphs} is then defined as follows.

\begin{algorithm}[h]
\caption{The density-greedy algorithm for general graphs}\label{alg:dg-graph}
\begin{algorithmic}[1]
\State $\pi \gets ()$
\State $\overbar{T} \gets (\{r\},\emptyset)$
\While{$p_{\contract{G}{\overbar{T}}}(\contract{G}{\overbar{T}}) \neq 0$}
    \State fix arbitrary min-max subtree $T^*$ of $\contract{G}{\overbar{T}}$ 
    \label{it:compute-tree-graph}
    \State denote by $T^+$ the extension of $T^*$ in $G$
    \label{it:compute-uncontracted-tree-graph}
    \State \textbf{run} \Cref{alg:dg-tree} on $\contract{T^+}{\overbar{T}}$ and obtain $\pi^+$
    \State \smash{append $\pi^+$ to $\pi$ and add edges of $T^+$ to $\overbar{T}$}
\EndWhile
\State \Return $\pi$
\end{algorithmic}
\label{alg:greedy-general-graphs}
\end{algorithm}

First, we show that we do not obtain the same competitive ratios for the density-greedy algorithm for general graphs as we did for trees.

\begin{restatable}{proposition}{lemgengraphsdensity}\label{lem:gen-graphs-density}
    The density-greedy algorithm is not $(\chi,\mu)$-competitive for any $\mu\geq 1$.
\end{restatable}
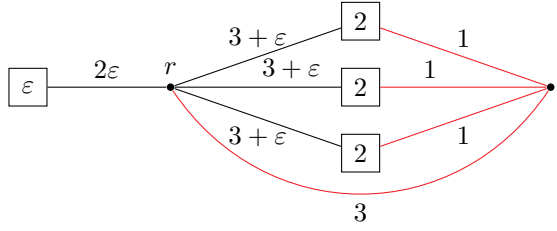
\begin{figure}
    \centering
   	\begin{tikzpicture}[scale=1.25]
     \tikzset{mininode/.style = {circle, fill, minimum size=2pt, inner sep=1pt},
        weightnode/.style = {rectangle, draw, minimum width=0.5cm,
                            minimum height = 0.5cm},
    }
        \node[mininode,label=above:{$r$}] at (0,0) (r) {};
	\node[mininode] at (4,0) (x) {};
	\node[weightnode] at (2,0.7) (A) {$2$};
	\node[weightnode] at (2,0) (B) {$2$};
        \node[weightnode] at (2,-0.7) (C) {$2$};
        \node[weightnode] at (-1.5,0) (E) {$\varepsilon$};

        \draw[Red] (r) .. controls (1,-1.5) and (3,-1.5) .. node[below]{\textcolor{black}{$3$}} (x);        
        \draw (r) -- node[above]{$2\varepsilon$} (E);
        \draw (r) -- node[above]{$3+\varepsilon$} (A);
        \draw (r) -- node[above,pos=0.7]{$3+\varepsilon$} (B);
        \draw (r) -- node[below]{$3+\varepsilon$} (C);
        \draw[Red] (x) -- node[above]{\textcolor{black}{$1$}} (A);
        \draw[Red] (x) -- node[above,pos=0.7]{\textcolor{black}{$1$}} (B);
        \draw[Red] (x) -- node[below]{\textcolor{black}{$1$}} (C);
	\end{tikzpicture}
    \caption{Graph used in the proof of \Cref{lem:gen-graphs-density}}
    \label{fig:gen-graphs-counterexmp}
\end{figure}
\begin{proof}
Consider the instance given in \Cref{fig:gen-graphs-counterexmp} for some $0<\varepsilon<\lfrac{1}{3}$. The eccentricity of the root vertex is $\chi=3+\varepsilon$.
The only min-max subtree of $G$ is tree $T$ shown in red.
By the definition of the algorithm, the four edges of $T$ are the first edges in the solution $\pi$, starting with the edge of cost $3$.
Thus, for any cost budget $B<4$, we have $p\bigl(\ALG(B)\bigr)=0$.
In particular, since $2\varepsilon+\chi<4$, we obtain
$\mu \,p\bigl(\ALG(2\varepsilon+\chi)\bigr)=0<\varepsilon=p\bigl(\OPT(2\varepsilon)\bigr)$ for all $\mu\geq 1$. Thus, the density-greedy algorithm for graphs is not $(\chi,\mu)$-competitive for any $\mu\geq 1$.
\end{proof}

We work towards showing that the density-greedy algorithm for general graphs is $(\lrp,2)$-competitive, where $\lrp$ denotes the maximum cost of a (vertex-disjoint) path in $G$ starting in the root $r$.
To do so, we again denote by $T_1$ the first min-max subtree that is computed in Step \ref{it:compute-tree-graph} of \Cref{alg:dg-graph}.
The following lemma is the counterpart to \Cref{lem:optimal_on_first_tree}.

\begin{restatable}{lemma}{lemgraphoptimalonfirsttree}
    \label{lem:graph_optimal_on_first_tree}
    For any cost budget $B\leq c(T_1)$, we have 
    $p\bigl(\ALG(B+\lrp)\bigr) \geq \dens(T_1) B$.
    Furthermore, if $B=c(T_1)$, we have
    $p\bigl(\ALG(B)\bigr) = \dens(T_1) B$.
\end{restatable}
\begin{proof}
    By the definition of the algorithm, the first $|T_1|$ edges of $\pi$ are exactly the edges of $T_1$.
    This yields $p\bigl(\ALG(c(T_1))\bigr)=\dens(T_1)c(T_1)$ proving the second statement of the lemma and the first one for all cost budgets $c(T_1)-\lrp\leq B\leq c(T_1)$.

    It remains to consider cost budgets $B<c(T_1)-\lrp$.
    Applying \Cref{lem:optimal_on_first_tree} on $T_1$, we obtain that
    $p\bigl(\ALG(B+\chi_1)\bigr)\geq \dens^*_1 B$,
    where $\chi_1$ is the eccentricity of the root in $T_1$ and $\dens^*_1$ is the maximum density of a rooted subtree in $T_1$, thus, $\dens^*_1 \geq d(T_1)$.
    The eccentricity $\chi_1$, however, cannot be compared to $\chi$.
    Instead, we use $\lrp$ as an upper bound on $\chi_1$ and obtain
    \begin{align*}
        p\bigl(\ALG(B+\lrp)\bigr)\geq p\bigl(\ALG(B+\chi_1)\bigr)\geq \dens_1^*B=d(T_1)B
    \end{align*}
    as claimed.   
\end{proof}

We continue by generalizing the notion of the density to subtrees of general graphs.
To this end, let~$T$ be a (not necessarily rooted) subtree of $G$.
Since the choice of the vertex~$r_T$ in~$T$ that has the shortest distance to $r$ is not necessarily unique when $G$ is a graph, we define the density of~$T$ over an underlying rooted subtree $T'\supset T$.
Then, the choice of the vertex~$r_T$ of $T$ that has the shortest distance to $r$ in $T'$ is again unique and we define the density of~$T$ over $T'$ as $\dens_{T'}(T)\coloneqq \bigl(p(T) - p(r_T)\bigr)\big/ c(T)$.
Thereby, this notion is consistent with the density used in \Cref{sec:greedy_trees}.
We are now ready to prove the upper bound of \Cref{thm:result_greedy_graphs}.

\thmresultgreedygraphs*

\begin{proof}
The fact that the density-greedy algorithm is not $(\chi,\mu)$-competitive for any $\mu \geq 1$ follows from \Cref{lem:gen-graphs-density}.

We proceed to show that $2p\bigl(\ALG(B+\lrp)\bigr)\geq p\bigl(\OPT(B)\bigr)$
    for every cost budget $B\in\R_{\geq0}$.
    We denote by $c(\pi)$ the total cost of all edges in the incremental solution $\pi$ computed by $\ALG$.
    If $c(\pi)\leq B$, we obtain that~$p(\ALG(B+\lrp))=p(G)\geq p\bigl(\OPT(B)\bigr)$ and thus, the claim holds.
    
    It remains to consider the case where $c(\pi)> B$.
    We denote by $T_1,\dots,T_n$ the extensions of the min-max subtrees computed in Step \ref{it:compute-uncontracted-tree-graph} of \Cref{alg:dg-graph} and let $k\in\{1,\dots,n\}$ be such that we have $\sum_{i=1}^{k-1}c(T_i)\leq B <\sum_{i=1}^k c(T_i)$.
    Note that by construction, $T_i$ and $T_j$ are edge disjoint for all $i\neq j$. Also, every $T_i$ is a tree since the min-max subtrees all have only one branch at the root.
    We denote by~$\overbar{T}_i=\bigcup_{j=1}^i T_j$ the union of all trees $T_1,\dots,T_i$ and set~$\overbar{T}_0=(\{r\},\emptyset)$.
    We claim that 
    \begin{align}
    \label{eq:density-decreasing}
    \dens_{\overbar{T}_k}(T_1)\geq \dens_{\overbar{T}_k}(T_2)\geq \dots \geq \dens_{\overbar{T}_k} (T_k).
    \end{align}
    To prove this, we assume for contradiction that there is some $1\leq i<k$ such that $\dens_{\overbar{T}_k}(T_i)<\dens_{\overbar{T}_k}(T_{i+1})$.
    However, since~$\contract{(T_i\cup T_{i+1})}{\overbar{T}_{i-1}}$ is a rooted tree in $\contract{G}{\overbar{T}_{i-1}}$, this assumption yields
    \begin{align*}
        \dens_{\contract{G}{\overbar{T}_{i-1}}}(\contract{(T_i\cup T_{i+1})}{\overbar{T}_{i-1}})
        =\dens_{\overbar{T}_{k}}(T_i\cup T_{i+1})
        >\dens_{\overbar{T}_{k}}(T_i)
        =\dens_{\contract{G}{\overbar{T}_{i-1}}}(\contract{T_i}{\overbar{T}_{i-1}}),
    \end{align*}
    contradicting the choice of $\contract{T_i}{\overbar{T}_{i-1}}$ as a min-max subtree in $\contract{G}{\overbar{T}_{i-1}}$. Thus, the claim holds.
    If we assume that $\OPT(B)\subseteq \overbar{T}_{k-1}$, using $\overbar{T}_{k-1}\subseteq \ALG(B)$ gives $p(\OPT(B))\leq p(\ALG(B))$ and the theorem follows immediately.
    Hence, we assume for the remaining part of this proof that $\OPT(B)\nsubseteq \overbar{T}_{k-1}$.
    Using the facts that $\contract{T_k}{\overbar{T}_{k-1}}$ is a min-max subtree of $\contract{G}{\overbar{T}_{k-1}}$ and that $\overbar{T}_{k-1}\subseteq \ALG(B)$, we obtain
    \begin{align}
        \dens_{\overbar{T}_k}(T_k)
        =\dens_{\contract{G}{\overbar{T}_{k-1}}}(\contract{T_k}{\overbar{T}_{k-1}}) \label{eq:density-estimation-1} 
        &\geq \frac{p_{\contract{G}{\overbar{T}_{k-1}}}(\contract{\OPT(B)}{\overbar{T}_{k-1}})}{c_{\contract{G}{\overbar{T}_{k-1}}}(\contract{\OPT(B)}{\overbar{T}_{k-1}})}\\
        &\geq\frac{p_{\contract{G}{\ALG(B)}}(\contract{\OPT(B)}{\ALG(B)})}{B}.\label{eq:density-estimation-2}
    \end{align}
    
Applying \Cref{lem:graph_optimal_on_first_tree} iteratively on $\contract{G}{\overbar{T}_{i-1}}$ with cost budget $c(T_i)$ for all $1 \leq i \leq k-1$ and for $\smash{\contract{G}{\overbar{T}_{k-1}}}$ with cost budget $\smash{B - c(\overbar{T}_{k-1})<c(T_k)}$, we obtain
    \begin{align*}
        p\bigl(\ALG(B+\lrp)\bigr)
        \geq \sum_{i=1}^{k-1}\dens_{\overbar{T}_k}(T_i)c(T_i)+\dens_{\overbar{T}_k}(T_k)\bigl(B-c(\overbar{T}_{k-1})\bigr)
        \geq \dens_{\overbar{T}_k}(T_k)B,
    \end{align*}
    where we used 
    $\dens_{\overbar{T}_k}(T_i)\geq\dens_{\overbar{T}_k}(T_{k})$ for all $i\leq k$.
    Finally, we conclude that
    \begin{align}
        p\bigl(\OPT(B)\bigr)
        &\leq p\bigl(\ALG(B)\bigr) + p_{\contract{G}{\ALG(B)}}\bigl(\contract{\OPT(B)}{\ALG(B)}\bigr)\notag\\
        &\leq p\bigl(\ALG(B)\bigr)+\dens_{\overbar{T}_k}(T_k)B\notag\\
        &\leq p\bigl(\ALG(B)\bigr)+p\bigl(\ALG(B+\lrp)\bigr) \label{eq:density-approximate}\\
        &\leq 2p\bigl(\ALG(B+\lrp)\bigr). \qedhere
    \end{align}
\end{proof}

The computational bottleneck of the density-greedy algorithm is the computation of a min-max subtree in Step~\ref{it:compute-tree-graph} of \Cref{alg:dg-graph}.
As the computation of a subtree of maximum density is $\mathsf{NP}$-hard (Lau et al.~\cite{LauNN06}), there is no polynomial implementation of this step, unless $\mathsf{P} = \mathsf{NP}$.
However, Hermans et al.~\cite{HermansLM22} give a polynomial $2$-approximation for the computation of a subtree with maximal density.
In order to obtain a polynomial algorithm for the incremental prize-collecting Steiner-tree problem on general graphs, we need to do the following adaptations to \Cref{alg:dg-graph}:
\begin{enumerate}
    \item In Step~\ref{it:compute-tree-graph}, we use the approximation algorithm of Hermans et al.~\cite{HermansLM22} to obtain a tree $T^*$ of $\contract{G}{\overbar{T}}$ with density $d_{\contract{G}{\overbar{T}}}(T^*) \geq d^*/2$ where $d^*$ is the maximum density of a rooted subtree in ${\contract{G}{\overbar{T}}}$.
    \item \label{it:merge-step} Let $T^*_i$ and $T^*_{i+1}$ be two such trees computed in two consecutive iterations of the while loop such that the density of $T_{i+1}^*$ is larger than the density of $T_i^*$. If both trees are rooted in the same vertex, we switch their order. Otherwise, we merge both trees.
\end{enumerate}

The second point of these adjustments makes sure that the densities of the subtrees computed by the algorithm are decreasing, i.e., that inequality \eqref{eq:density-decreasing} in the proof of \Cref{thm:result_greedy_graphs} still holds. We observe that when running \Cref{alg:dg-graph} with these adjustments, the proof of \Cref{thm:result_greedy_graphs} goes through except for the fact that the right hand sides of inequalities \eqref{eq:density-estimation-1} and~\eqref{eq:density-estimation-2} are multiplied with $1/2$. 
This ultimately leads to an additional loss of $p\bigl(\ALG(B+\gamma)\bigr)$ in~\eqref{eq:density-approximate}. 
We obtain the following result.

\corresultgreedygraphs*

\subsection{A Capacity-Scaling Algorithm}

The density-greedy algorithm turns out to be a $(\lrp,2)$-competitive algorithm for the prize-collecting Steiner-tree problem on general graphs.
This result is particularly strong when the given instance $G$ is sparse and hence, the maximum cost of a (vertex-disjoint) path starting in the root is small compared to the cost of a rooted spanning subtree of $G$.
However, in general, the value $\lrp$ can be large.
This applies for example, when $G$ is Hamiltonian.
In such cases, it is more desirable to obtain an $(\alpha,\mu)$-competitive algorithm where $\alpha$ depends on $\chi$ rather than on $\lrp$.
To this end, we introduce the \emph{capacity-scaling algorithm} inspired by the incremental maximization algorithm of~Bernstein et al.~\cite{BernsteinD0H22}.
The basic idea is as follows.
At the beginning the solution $\pi$ is set to the empty sequence $()$.
Then, in each iteration~$i=0,1,\dots$ the algorithm computes a maximum total prize rooted subtree $T_i$ of $G$ that obeys the cost budget $2^i\chi$.
Afterwards, we run \Cref{alg:dg-tree} on $T_i$ and append the returned solution to $\pi$.
This is done in such a way, that the solution $\pi$ contains neither double or parallel edges nor cycles.

The \emph{capacity-scaling algorithm} is given in more detail in \Cref{alg:cap}.

\begin{algorithm}
\caption{The capacity-scaling algorithm}\label{alg:cap}
\begin{algorithmic}[1]
\State $i \gets 0$
\State $\pi \gets ()$
\While{$\pi$ does not span all vertices in $V^*$}
    \State $T_i \in \argmax_{}\{p(T) : T\in \T, c(T)\leq 2^i\chi\}$\label{it:solve-budget-problem}
    \State \textbf{run} \Cref{alg:dg-tree} on $T_i$ with root $r$ and obtain $\pi_i$
    \State append $\pi_i$ to $\pi$ (skip double edges, parallel edges, and edges that close cycles)
    \State $i \gets i+1$
\EndWhile
\State \Return $\pi$
\end{algorithmic}
\end{algorithm}

In order to analyse the capacity-scaling algorithm properly, we give the following lemmata for general trees (not necessary subtrees of $G$).
Similar to the concept of min-max subtrees, we introduce the definition of min-max trees.
In particular, we call a tree~$T$ with root vertex~$r_T$ a \emph{min-max tree with root~$r_T$} if $p(r_T)=0$ and for all rooted subtrees $T'=(V_{T'},E_{T'}) \subset T$ with~$r_T \in V_{T'}$ and $c(T')>0$, we have $\dens(T) > \dens(T')$.
First, we generalize \Cref{lem:higher_density_in_induced_subtrees} to min-max trees.
The only adaption in the proof is that we use that $T$ is a min-max tree instead of a min-max subtree.
Specifically, we obtain the following result.

\begin{lemma}\label{lem:higher_density_in_induced_subtrees_graphs}
    Let $T$ be a min-max tree with root $r_T$. Then, for every edge $e$ in $T$ we have $\dens_T(T_e) \geq \dens(T)$, where $T_e$ is the branch of $T$ rooted at $e$.
    Strict inequality holds, whenever $T\neq T_e$.
\end{lemma}

The following lemma is crucial to proof \Cref{lem:lambda_portion}, where \Cref{lem:lambda_portion_min_max_arb} is used for $\delta=\left(1-2^{-k+1}\right)$.

\begin{restatable}{lemma}{lemlambdaprotionminmaxarb}\label{lem:lambda_portion_min_max_arb}
    Let $T$ be a tree with root $r_T$ and let $\delta\in[0,1]$ and $k\in\N$.   
    Assume that for every $\lambda' \in[0,1]$ and for every min-max tree $T'$ with root $r_{T'}$, there exists a forest $F'=(V_{F'},E_{F'})\subseteq T'$ with $r_{T'}\in V_{F'}$ and at most $k$ components such that
    $c(F') \leq \lambda' c(T')$ and $p(F') \geq \delta \lambda' p(T')$.
    Then, for all $\lambda \in[0,1]$, there exists a forest $F=(V_F,E_F)\subseteq T$ with $r_T\in V_F$ and at most~$k$ components such that
    $c(F) \leq \lambda c(T)$ and $p(F) \geq \delta \lambda p(T)$.
\end{restatable}
\begin{proof}
    Since $r_T$ needs to be contained in both $F$ and $T$ and $\delta,\lambda\leq 1$, we may assume without loss of generality that $p(r_T)=0$.
    Analogously to the density-greedy algorithm, we recursively define a sequence of min-max subtrees of $T$ as follows.
    Let $T_1^*$ be a fixed min-max subtree of~$T$.
    We denote by $T_i\subseteq T$ the extension of $T_i^*$. Note that $T_1^*=T_1$.
    For all $i>1$, we obtain~$T_i^*$ by fixing an arbitrary min-max subtree of $\contract{T}{\overbar{T}_{i-1}}$, where $\overbar{T}_{i-1}$ is defined by~$\overbar{T}_{i}\coloneqq \bigcup_{j=1}^i T_j$ with $\overbar{T}_0=(\{r_T\},\emptyset)$.
    By \Cref{lem:alpern_and_lidbetter} (\ref{item:alpern3}), the min-max subtree $T_i^*$ has only one branch at the root vertex $r_T$ of $\contract{T}{\overbar{T}_{i-1}}$.
    Thus, every extension $T_i$ is a connected subgraph of $T$.
    By construction, we obtain $\dens_T(T_i)\leq \dens(\overbar{T}_i)\leq \dens(\overbar{T}_{i-1})$. Let $r_i$ be the vertex contained in~$T_i$ that has shortest distance to $r_T$ in $T$.

    We may assume without loss of generality that $T$ has no leafs with prize $0$ as we could otherwise just consider $T$ without those vertices.
    Then, there is some finite $n\in \N$ such that~$T=\overbar{T}_n$ and $\dens_T(T_i)>0$ for all $1\leq i\leq n$. Recall that the density of~$T_i$ over~$T$ is given by~$\dens_T(T_i)=\bigl(p(T_i)-p(r_i)\bigr)\big/c(T_i)$.
    Hence, we can obtain the following bound on the density of $T$ for every $1\leq i\leq n$ and every $C\leq c(T_i)$.
    \begin{multline}\label{eq:densities_min_max_trees}
        \dens(T)
        =\dens(\overbar{T}_n)
        \leq \dens(\overbar{T}_i)
        =\frac{p(\overbar{T}_i)}{c(\overbar{T}_i)}
        =\frac{p(\overbar{T}_{i-1})+p(T_i)-p(r_i)}
        {c(\overbar{T}_{i-1})+c(T_i)} \\
        =\frac{\dens(\overbar{T}_{i-1})c(\overbar{T}_{i-1})+\dens_T(T_i)c(T_i)}
        {c(\overbar{T}_{i-1})+c(T_i)}
        \leq \frac{\dens(\overbar{T}_{i-1})c(\overbar{T}_{i-1})+\dens_T(T_i)C}
        {c(\overbar{T}_{i-1})+C},
    \end{multline}
    where for the last inequality we used the fact that $\dens_T(T_i)\leq \dens(\overbar{T}_{i-1})$.

    Let $j\coloneqq \argmin\{1\leq i\leq n : \lambda c(T)\leq c(\overbar{T}_i)\}$.
    Then, we have $0\leq \lambda c(T)-c(\overbar{T}_{j-1})\leq c(T_j)$, which yields
    $\lambda'\coloneqq \bigl(\lambda c(T)-c(\overbar{T}_{j-1})\bigr)\big/c(T_j)\in [0,1]$.

    Since $T_j^*$ is a min-max subtree of $\contract{T}{\overbar{T}_{j-1}}$, $T_j$ is a min-max tree with root $r_j$ for the prize function~$p'(r_j)=0$ and $p'(v)=p(v)$, else.
    Thus, we may fix an arbitrary $\delta\in[0,1]$ and apply the assumption of the lemma. 
    As a result, we obtain a forest $F_j\subseteq T_j$ that contains $r_j$, consists of at most $k$ components, has a total cost of at most
    $c(F_j)\leq \lambda'c(T_j)=\lambda c(T)-c(\overbar{T}_{j-1})$ and collects a total prize of at least~$p(F_j)-p(r_j)=p'(F_j)\geq \delta \lambda' p'(T_j)=\delta \lambda' (p(T_j)-p(r_j))$.
    Since~$r_j$ is also contained in~$\overbar{T}_{j-1}$, we can set $F\coloneqq F_j\cup \overbar{T}_{j-1}$ without gaining another component.
    Further, $F$ has a total cost of at most $c(F)\leq c(F_j)+c(\overbar{T}_{j-1})\leq \lambda c(T)$.
    To conclude the proof, it remains to show that $F$ collects at least a total prize of $\delta \lambda p(T)$.
    \begin{align*}
        p(F)
        &= p(\overbar{T}_{j-1}) + p(F_j) - p(r_j) \\
        &\geq p(\overbar{T}_{j-1}) + \delta \lambda' (p(T_j) - p(r_j))\\
        &= p(\overbar{T}_{j-1}) + \delta \frac{\lambda c(T)-c(\overbar{T}_{j-1})}{c(T_j)} (p(T_j) - p(r_j))\\
        &\geq \delta\left[ \dens(\overbar{T}_{j-1})c(\overbar{T}_{j-1}) + \frac{p(T_j) - p(r_j)}{c(T_j)} \bigl(\lambda c(T)-c(\overbar{T}_{j-1})\bigr) \right]
    \end{align*}
    Letting $C\coloneqq \lambda c(T)-c(\overbar{T}_{j-1})\leq c(T_j)$ and using \eqref{eq:densities_min_max_trees}, we conclude
    \begin{align*}
        p(F)
        &\geq \delta\left[ \dens(\overbar{T}_{j-1})c(\overbar{T}_{j-1}) + \dens_T(T_j) C \right]\\
        &\geq \delta \dens(T) \bigl(c(\overbar{T}_{j-1})+C\bigr)\\
        &= \delta \lambda \dens(T) c(T)\\
        &= \delta \lambda p(T). \qedhere
    \end{align*}
\end{proof}

Using the results of \Cref{lem:lambda_portion_min_max_arb}, we can show the following.
\pagebreak[4]

\begin{restatable}{lemma}{lemlambdaportion}\label{lem:lambda_portion}
    Let $T$ be a tree with root $r_T$ and let $\lambda\in[0,1]$ and $k\in \N$.
    Then there exists a forest $F=(V_F,E_F)\subseteq T$ with $r_T\in V_F$ and at most $k$ components such that
    \begin{align*}
        c(F) \leq \lambda c(T) \quad \text{and} \quad p(F) \geq \left(1-2^{-k+1}\right) \lambda p(T).
    \end{align*}
\end{restatable}
\begin{proof}
    Since $r_T$ needs to be contained in both $F$ and $T$ and $\lambda\leq 1$, we may assume without loss of generality that $p(r_T)=0$.
    Note that for $c(T)=0$ or $\lambda=1$ setting $F=T$ fulfills the desired properties.
    Thus, we may assume for the remaining part of this proof that $\lambda<1$ and~$c(T)>0$.

    We prove the lemma via induction on $k$. For $k=1$, the forest $F=(\{r_T\},\emptyset)$ is the trivial solution.
    For the induction step we assume that for an arbitrary but fixed integer $k-1\geq 1$ the statement holds true.
    By \Cref{lem:lambda_portion_min_max_arb}, it suffices to prove the lemma for $k$ under the assumption that $T=(V_T,E_T)$ is a min-max tree with root $r_T$.
    To this end, we give each edge $e\in E_T$ a direction to find the forest $F$ in $T$ that fulfills the desired properties.
    For this purpose, let $e=\{u,v\}$ be an edge of $T$ and assume that $u$ is the vertex that has the shorter distance to $r_T$ in $T$.
    Recall that by $T_e$ we denote the branch of $T$ rooted at $e$.
    We direct~$e$ from~$v$ to~$u$ (i.e., towards $r_T$) if $c(T_e)\leq \lambda c(T)$ and from $u$ to $v$ else.
    Since~$T$ is a tree, it follows that there exists a vertex $v\in V_T$ that has only incoming edges.
    Let~$T_v\subseteq T$ be the branch of~$T$ rooted in $v$. We distinguish two cases.

    First, assume $c(T_v)\leq \lambda c(T)$.
    We claim that the forest $F\coloneqq T_v\cup \{r_T\}$ fulfills the desired properties.
    Since $\lambda<1$ and $T_{r_T}=T$, we have $c(T_v)<c(T)$ and hence, $v\neq r_T$.
    Thus, there is a unique edge~$e=\{u,v\}$ such that $u$ lies on the unique $r_T$-$v$-path in $T$.
    Since~$v$ has only incoming edges, $e$ is directed from~$u$ to~$v$ and, therefore, $c(T_e)>\lambda c(T)$ must hold.
    With~$T$ being a min-max tree, we can apply \Cref{lem:higher_density_in_induced_subtrees_graphs} and obtain $\dens_T(T_e)\geq \dens_T(T)$.
    Together with~$p(T_v)=c(T_e)\dens_T(T_e)$, it follows that
    \begin{align*}
        p(F)
        =p(T_v)
        =c(T_e)\dens_T(T_e)
        >\lambda c(T) \dens_T(T)
        =\lambda p(T)
        \geq \left(1-2^{-\kappa+1}\right) \lambda p(T).
    \end{align*}
    Since $T_v$ is connected, the forest $F$ has at most two components and fulfills the desired properties.

    \begin{figure}
        \centering
        \begin{subfigure}[c]{0.49\textwidth}
            \centering
            \begin{tikzpicture}
                \tikzset{mininode/.style = {circle, fill, minimum size=2pt, inner sep=1pt},
                    weightnode/.style = {rectangle, draw, minimum width=0.5cm,
                                        minimum height = 0.5cm},
                }
                \node[mininode,label=above:{$r_T$}] at (-1.25,1.25) (r) {};
                \node[mininode,label=right:$v$] at (0,0) (v) {};
                \node[mininode] at (-1.5,-1) (v1) {};
                \node[mininode] at (-0.66,-1.25) (vl-1) {};
                \node[mininode] at (0.66,-1.25) (vl) {};
                \node[mininode] at (1.5,-1) (vn) {};
        
                \draw (r) -- (-1.5,0.95);
                \draw (r) -- (-0.95,0.95);
                \draw (-0.5,0.5) -- (v);
                \draw[BurntOrange] (v) --node[above left=-1mm] {$U$} (v1);
                \draw[BurntOrange] (v) --node[left, near end] {$\dots$} (vl-1);
                \draw[BurntOrange] (v) --node[above right=-1.2mm, pos=0.9] {$e^*$} (vl);
                \draw (v) -- (vn);
                \draw (v1) -- (-1.8,-1.4);
                \draw (v1) -- (-1.4,-1.4);
                \node (dots1) at (-1.6,-1.55) {$\dots$};
                \draw (vl-1) -- (-0.96,-1.6);
                \draw (vl-1) -- (-0.36,-1.6);
                \node (dots2) at (-0.66,-1.75) {$\dots$};
                \draw (vl) -- (0.4,-1.6);
                \draw (vl) -- (0.75,-1.6);
                \draw (vl) -- (1,-1.6);
                \node (dots3) at (0.7,-1.75) {$\dots$};
                \draw (vn) --  node[at end, below=1pt] {$\dots$} (1.8,-1.4);
        
                \node[label=below:{\dots}] at (-1.6,0.95) {}; 
                \node[label=below:{\dots}] at (-0.7,0.95) {}; 
                
                \draw[rounded corners, Red] (0.7,0.3) -- (-0.2,-2) -- (-2.2,-2) --node[above left] {$T_{-e^*}$} (-1.7,-0.3) -- (-0.3,0.3) -- cycle;
                \draw[rounded corners, MidnightBlue] (0.7,0.2) -- (1.3,-2) --node[below] {$T_{e^*}$} (0.3,-2) -- (-0.2,0.2) -- cycle;
            \end{tikzpicture}
            \caption{Finding subtrees $T_{e^*}$ and $T_{-e^*}$}
            \label{fig:forest_more_components_trees_for_recursion}
        \end{subfigure}
        \hfill
        \begin{subfigure}[c]{0.49\textwidth}
            \centering
            \begin{tikzpicture}
                \tikzset{mininode/.style = {circle, fill, minimum size=2pt, inner sep=1pt},
                    weightnode/.style = {rectangle, draw, minimum width=0.5cm,
                                        minimum height = 0.5cm},
                }
                \node[mininode,label=above:{$r_T$}] at (-1.25,1.25) (r) {};
                \node[mininode,label=right:$v$] at (0,0) (v) {};
                \node[mininode] at (-1.5,-1) (v1) {};
                \node[mininode] at (-0.66,-1.25) (vl-1) {};
                \node[mininode] at (0.66,-1.25) (vl) {};
                \node[mininode] at (1.5,-1) (vn) {};
        
                \draw (r) -- (-1.5,0.95);
                \draw (r) -- (-0.95,0.95);
                \draw (-0.5,0.5) -- (v);
                \draw (v) -- (v1);
                \draw (v) -- node[left, near end] {$\dots$} (vl-1);
                \draw (v) -- (vl);
                \node [above=9.5pt of vl.north] {$\dots$};
                \draw (v) -- (vn);
                \draw (v1) -- (-1.8,-1.4);
                \draw (v1) -- (-1.4,-1.4);
                \node (dots1) at (-1.6,-1.55) {$\dots$};
                \draw (vl-1) -- (-0.96,-1.6);
                \draw (vl-1) -- (-0.36,-1.6);
                \node (dots2) at (-0.66,-1.75) {$\dots$};
                \draw (vl) -- (0.3,-1.8);
                \node[mininode] at (0.3,-1.8) (t) {};
                \draw (t) -- (0.1,-2.1);
                \draw (t) -- (0.4,-2.1);
                \node (dots5) at (0.25,-2.3) {$\dots$};
                \draw (vl) -- (0.7,-1.7);
                \node (dots4) at (0.8,-1.9) {$\dots$};
                \draw (vl) -- (1.1,-1.7) -- (1.2,-1.8);
                \node[mininode] at (1.1,-1.7) {};
                \node (dots3) at (1.3,-2) {$\dots$};
                \draw (vn) --  node[at end, below=1pt] {$\dots$} (1.8,-1.4);
        
                \node[label=below:{\dots}] at (-1.6,0.95) {}; 
                \node[label=below:{\dots}] at (-0.7,0.95) {}; 
                
                \draw[rounded corners, Red] (0.6,0.2) -- (-0.2,-2) -- (-2.2,-2) --node[above left] {$T_{-e^*}$} (-1.7,-0.3) -- (-0.3,0.2) -- cycle;
                \draw[ForestGreen] (-1.25,1.4) circle [radius=0.4];
                \draw[MidnightBlue] (0.15,-0.1) circle [radius=0.3];
                \draw[MidnightBlue] (1.2,-1.9) circle [radius=0.3];
                \draw[rounded corners, MidnightBlue] (0.35,-1.6) -- (0.6,-2.4) node[right] {$F^*$} -- (-0.2,-2.4) -- cycle;
            \end{tikzpicture}
            \caption{The forest $F$ is given by $F=F^*\cup T_{-e^*}\cup \{r_T\}$}
            \label{fig:forest_with_more_components_recursion}
        \end{subfigure}
        \caption{Recursion to find forest $F$ in the proof of \Cref{lem:lambda_portion}}
        \label{fig:forest_with_more_components}
    \end{figure}
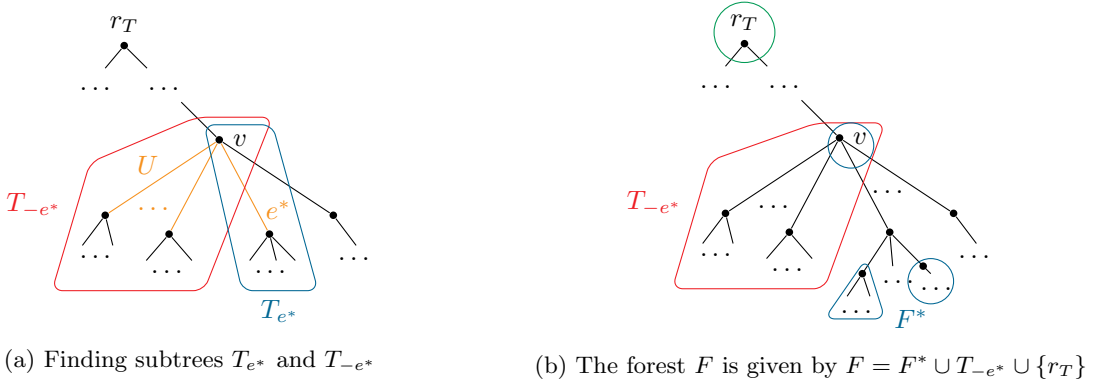

    It remains to consider the case where $c(T_v) >\lambda c(T)$.
    Let $e_1,\dots,e_j$ be the edges of $T_v$ that are incident to $v$.
    Since all edges $e_1,\dots,e_j$ are directed to $v$, $c(T_{e_i})\leq \lambda c(T)$ holds for all~$1\leq i\leq j$.
    Let~$U\subseteq \{e_1,\dots,e_j\}$ be inclusion-wise minimal such that $\lambda c(T)< \sum_{e\in U} c(T_{e})$.
    In particular, $|U|\geq 2$ since $c(T_{e_i})\leq \lambda c(T)$ for all $1\leq i\leq j$.
    Hence, there exists an edge~$e^*\in U$ with $c(T_{e^*})\leq \sum_{e\in U} c(T_{e})/2$.
    We denote by $T_{-e^*}\coloneqq \bigcup_{e\in U \setminus \{e^*\}}T_{e}$ (see \Cref{fig:forest_more_components_trees_for_recursion}).
    Then, we have $\sum_{e\in U} c(T_{e})=c(T_{e^*})+c(T_{-e^*})$ and therefore $c(T_{-e^*})\geq \sum_{e\in U} c(T_{e})/2>\lambda c(T)/2$.
    Combining the upper and lower bounds for $c(T_{-e^*})$, we obtain
    $\lambda c(T)/2 < c(T_{-e^*}) \leq \lambda c(T)$.
    We denote by $\lambda' \coloneqq \bigl(\lambda c(T)-c(T_{-e^*})\bigr)\big/c(T_{e^*})$.
    With $c(T_{-e^*})\leq \lambda c(T)$ and $\lambda c(T)-c(T_{-e^*})<\sum_{e\in U}c(T_e)-c(T_{-e^*})=c(T_{e^*})$, it follows that $\lambda'\in [0,1]$.
    Thus, we may apply the induction hypothesis for $k-1$ and $\lambda'$ on the tree $T_{e^*}$ with root vertex $v$. Thereby, we obtain a forest~$F^*=(V_{F^*},E_{F^*})\subseteq T_{e^*}$ with $v\in V_{F^*}$ and at most $k-1$ components such that
    \begin{align*}
        c(F^*) \leq \lambda' c(T_{e^*}) \quad \text{and} \quad p(F^*) \geq \left(1-2^{-k+2}\right) \lambda' p(T_{e^*}).
    \end{align*}
    We set $F\coloneqq F^*\cup T_{-e^*}\cup \{r_T\}$ as depicted in \Cref{fig:forest_with_more_components_recursion} and show that this forest fulfills the desired properties.

    To this end, first note that one component of $F^*$ contains the root $v$ of both $T_{e^*}$ and~$T_{-e^*}$ and thus,~$F$ has at most $k$ components and contains $r_T$.
    Further, the total cost of $F$ can be upper bounded by
    \begin{align*}
        c(F)
        =c(F^*)+c(T_{-e^*})
        \leq \lambda' c(T_{e^*}) + c(T_{-e^*})
        =\lambda c(T).
    \end{align*}
    It remains to show that $p(F)\geq \left(1-2^{-k+1}\right) \lambda p(T)$.
    In this regard, recall that we assumed $T$ to be a min-max tree rooted at $r_T$ and thus, we have $\dens_T(T_e)\geq \dens_T(T)$ for every edge $e\in E_T$ by \Cref{lem:higher_density_in_induced_subtrees}. In particular, this shows $\dens_T(T_{e^*})\geq \dens_T(T)$ and $\dens_T(T_{-e^*})\geq \dens_T(T)$.
    Combining this with all previous observations, we obtain
    \begin{align*}
        p(F)
        &=p(F^*)+p(T_{-e^*})-p(v)\\
        &\geq \left(1-2^{-k+2}\right) \lambda' p(T_{e^*}) +p(T_{-e^*})-p(v)\\
        &\geq \left(1-2^{-k+2}\right) \frac{\lambda c(T)-c(T_{-e^*})}{c(T_{e^*})} \dens_T(T_{e^*})c(T_{e^*}) + \dens_T(T_{-e^*})c(T_{-e^*})\\
        &\geq \dens_T(T) \left[ \left(1-2^{-k+2}\right) \bigl(\lambda c(T)-c(T_{-e^*})\bigr) + c(T_{-e^*}) \right]\\
        &= \dens_T(T)\bigl[ \left(1-2^{-k+2}\right)\lambda c(T) +2^{-k+2}c(T_{-e^*}) \bigr] \\
        & > \dens_T(T)\bigl[ (1-2^{-k+2})\lambda c(T) +2^{-k+1}\lambda c(T) \bigr]\\
        &= (1+2^{-k+1})\lambda p(T).
    \end{align*}
    We conclude that there exists a forest $F$ with the desired properties whenever $T$ is a min-max tree with root $r_T$. As we have argued above, using \Cref{lem:lambda_portion_min_max_arb} the claim then also follows for any tree $T$ rooted in~$r_T$. This finishes the proof.
\end{proof}

The following corollary will be useful for the analysis of the capacity-scaling algorithm.

\begin{restatable}{corollary}{corscalingopt}\label{cor:scaling_opt}
    For every $\delta \geq 1$, $h \in \N$, and every cost budget $B\in\R_{\geq 0}$, we have that~$p(\OPT(B+h \chi))\geq (1-2^{-h})\delta^{-1} p(\OPT(\delta B))$.
\end{restatable}
\begin{proof}
    We apply \Cref{lem:lambda_portion} to $T=\OPT(\delta B)$ with root $r$, $\lambda=\delta^{-1}\in[0,1]$ and $k=h+1$ components and obtain a forest~$F=(V_F,E_F)$ with $r\in V_F$, at most $h+1$ components, a total cost of at most $c(F) \leq \delta^{-1} c(\OPT(\delta B)) \leq B$ and a total collected prize of at least $p(F) \geq (1-2^{-h})\delta^{-1} p(\OPT(\delta B))$.
    Since every component of $F$ can be connected to the root with cost at most $\chi$, there exists a rooted subtree $T^*\subseteq T$ with cost at most $c(T^*) \leq c(F) + h\chi \leq B + h\chi$, that collects at least a total prize of $p(T^*)\geq p(F) \geq (1-2^{-h})\delta^{-1} p(\OPT(\delta B))$. This yields the lower bound on the collected prize of $\OPT(B+h\chi)$ claimed in the corollary.
\end{proof}

We are now ready to prove the main theorem of this subsection.

\mainthmscalingmoreheadstart*
\begin{proof}
    For cost budget $B\in\R_{\geq 0}$, let $k\in \N$ be such that $B\in \left[ \left(2^{k+1}-4\ell\right)\chi, \left(2^{k+2}-4\ell\right)\chi \right)$.
    We denote by~$\ALG(B)$ the solution of the capacity-scaling algorithm for cost budget $B$.
    Since $\ALG$ grows monotonously, it holds that 
    \begin{align} \label{eq:scaling_alg_analysis_1}
        p(\ALG(B+(4\ell-1)\chi))
        &\geq p(\ALG((2^{k+1}-1)\chi)).
    \end{align}
    Let $T_0, T_1, \dots$ be the subtrees of $G$ computed by the capacity-scaling algorithm (\Cref{alg:cap}).
    Note that 
    \begin{align*}
        \sum_{i=0}^{k} c(T_i)
        \leq \sum_{i=0}^{k} 2^i \chi
        = (2^{k+1}-1) \chi.
    \end{align*}
    Thus, $\ALG((2^{k+1}-1) \chi)$ contains the tree $T_k$, i.e.,
    \begin{align} \label{eq:scaling_alg_analysis_2}
        p(\ALG((2^{k+1}-1)\chi))
        &\geq p(T_{k}).
    \end{align}
    Let $\OPT(2^k\chi)$ be the optimal solution in~$G$ for cost budget $2^k\chi$.
    Then, a spanning tree of~$\OPT(2^k\chi)$ is a potential candidate for the capacity-scaling algorithm when computing tree~$T_k$, which thus yields $p(\OPT(2^k\chi))= p(T_k)$.
    Finally, we use \Cref{cor:scaling_opt} with $\delta=4$, $h=\ell$ and cost budget $B_k\coloneqq(2^k-\ell)\chi \geq \frac{B}{4}$ and together with inequalities $(\ref{eq:scaling_alg_analysis_1})$ and $(\ref{eq:scaling_alg_analysis_2})$ we obtain the claim as follows
    \begin{align}
        p(\ALG(B+(4\ell-1)\chi))
        &\geq p(\ALG((2^{k+1}-1)\chi)) \notag\\
        &\geq p(T_{k}) \notag\\
        &= p(\OPT(2^k\chi))\label{eq:capacity}\\
        &= p(\OPT(B_k+\ell \chi))\notag\\
        &\geq \frac{1-2^{-\ell}}{4}p\bigl(\OPT(4 B_k)\bigr)\notag\\
        &\geq \frac{2^\ell-1}{2^{\ell+2}} p\bigl(\OPT(B)\bigr). \notag \qedhere
    \end{align}
\end{proof}

The computational bottleneck of the capacity-scaling algorithm is the solution of the prize-collecting Steiner tree problem with budget $2^i \chi$ in Step~\ref{it:solve-budget-problem} of \Cref{alg:cap}. 
The $\mathsf{NP}$-hardness of the Steiner tree problem implies the hardness of the budget variant of the prize-collecting Steiner tree problem, so there is no polynomial implementation of Step~\ref{it:solve-budget-problem}, unless $\mathsf{P}=\mathsf{NP}$.
In terms of approximations, only a quasi-polynomial algorithm with a poly-logarithmic approximation guarantee is known (Ghuge and Nagarajan~\cite{GhugeN22}); see also the discussion by Paul et al.~\cite{Paul23erraturm}. We leave it as an open problem to devise a polynomial algorithm that is $(O(1)\chi, O(1))$-competitive.

\subsection{Lower Bound for General Graphs}\label{sec:lb-gg}

We complement our algorithmic results with a lower bound for general graphs.
For the proof, we first construct a small graph for which there is a trade-off between the maximal prize that can be collected at budgets $3$ and $7$, respectively. A careful analysis of the graph that has an arbitrary large number of copies of this graph joined at the root then yields the following result.
	
\thmlowerboundgraph*
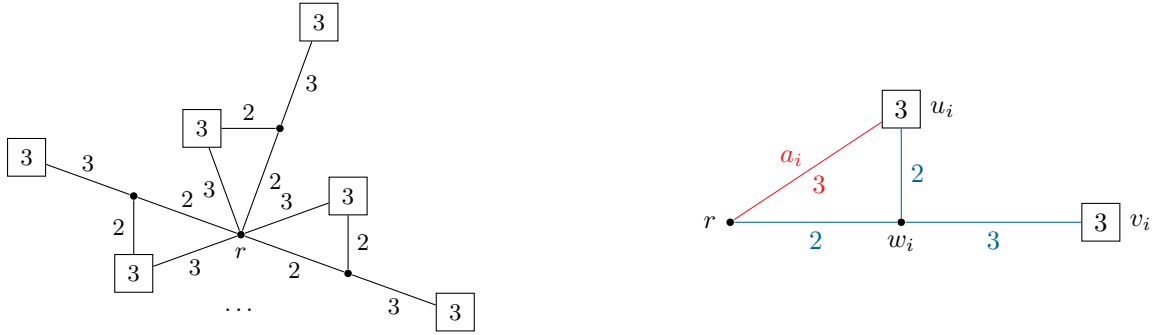
\begin{figure}
        \centering
        \begin{subfigure}[c]{0.49\textwidth}
            \begin{tikzpicture}\small
             \tikzset{mininode/.style = {circle, fill, minimum size=2pt, inner sep=1pt},
                weightnode/.style = {rectangle, draw, minimum width=0.5cm,
                                    minimum height = 0.5cm},
            }
                \node[mininode,label=below:{$r$}] at (0,0) (r) {};
                \node[mininode] at (-20:1.5) (A1) {};
                \node[weightnode] at (20:1.5) (B1) {$3$};
                \node[weightnode] at (-20:3) (C1) {$3$};
                \draw (r) -- node[below]{$2$} (A1);
                \draw (r) -- node[above]{$3$} (B1);
                \draw (A1) -- node[right]{$2$} (B1);
                \draw (A1) -- node[below] {$3$} (C1);
                
                \node[mininode] at (70:1.5) (A2) {};
                \node[weightnode] at (110:1.5) (B2) {$3$};
                \node[weightnode] at (70:3) (C2) {$3$};
                \draw (r) -- node[right]{$2$} (A2);
                \draw (r) -- node[left]{$3$} (B2);
                \draw (A2) -- node[above]{$2$} (B2);
                \draw (A2) -- node[right] {$3$} (C2);
    
                \node[mininode] at (160:1.5) (A3) {};
                \node[weightnode] at (200:1.5) (B3) {$3$};
                \node[weightnode] at (160:3) (C3) {$3$};
                \draw (r) -- node[above]{$2$} (A3);
                \draw (r) -- node[below]{$3$} (B3);
                \draw (A3) -- node[left]{$2$} (B3);
                \draw (A3) -- node[above] {$3$} (C3);
    
                \node at (270:1) (dots) {$\dots$};
            \end{tikzpicture}
            \caption{The instance $G$ consists of $n$ copies of the graph~$G_i$ (on the right) that are all merged in the root vertex~$r$.}
            \label{fig:lowerbound-generalgraphs-complete}
        \end{subfigure}
        \hfill
        \begin{subfigure}[c]{0.49\textwidth}
            \centering
            \begin{tikzpicture}[scale=1.5]
                \tikzset{mininode/.style = {circle, fill, minimum size=2pt, inner sep=1pt},
                        weightnode/.style = {rectangle, draw, minimum width=0.5cm,
                                                minimum height = 0.5cm},
                        fakenode/.style = {rectangle, minimum width=0.5cm,
                                                minimum height = 0.5cm},
                        }
                \useasboundingbox (-0.325,-0.875) rectangle (3.75,1.85);
                \node[mininode, label=left:{$r$}] at (0,0) (r) {};
                \node[weightnode, label=right:{$u_i$}] at (1.5,1) (A) {3};
                \node[weightnode, label=right:{$v_i$}] at (3.25,0) (B) {3};
                \node[mininode, label=below:{$w_i$}] at (1.5,0) (C) {};
                
                \draw[Red] (r) -- node[above left=-0.25em]{$a_i$} node[below right=-0.25em]{$3$} (A);
                \draw[MidnightBlue] (r) -- node[below]{$2$} (C);
                \draw[MidnightBlue] (C) -- node[below]{$3$} (B);
                \draw[MidnightBlue] (C) -- node[right]{$2$} (A);
            \end{tikzpicture}
            \caption{
            The graph $G_i$ with $\OPT(3)=A_i$ shown in red, and with $\OPT(7)=B_i$ shown in blue.}
            \label{fig:lowerbound-generalgraphs-component}
        \end{subfigure}
        \caption{Instance with no $(1,\alpha)$-competitive incremental solution.}
        \label{fig:lowerbound-generalgraphs}
\end{figure}
\begin{proof}
    We consider the instance $G$ that is given by $n\in\N$ many copies $G_1,\dots,G_n$ of the graph 
    that are all merged in the root vertex $r$ shown in \Cref{fig:lowerbound-generalgraphs-complete}.
    We denote by $A_i$ the subtree of $G_i$ that only consists of the edge $a_i$, i.e., the red tree in \Cref{fig:lowerbound-generalgraphs-component} and by $B_i$ the complement of $A_i$ in $G_i$, i.e., the blue tree in \Cref{fig:lowerbound-generalgraphs-component}.
    In particular, for $n=1$ we have $\OPT(3)=A_1$ and $\OPT(7)=B_1$.

    Let $\ALG$ be an $(\alpha,\mu)$-competitve algorithm for the incremental prize-collecting Steiner-tree on general graphs, where $\alpha$ only depends on $\lrp$.
    Let $T=(V_T,E_T)\coloneqq \ALG(3n)$.
    To show that $\mu\geq 17/16$, we define~$k_n\coloneqq \abs{\{1\leq i \leq n : a_i \in E_T \}}$ to be the number of edges of type $a$ in $\ALG(3n)$.
    Based on this value we distinguish two cases.
    
    First, assume $k_n\leq 10n/17$.
    We denote by $T_i\coloneqq G_i \cap \ALG(3n)$ the rooted tree in the $i$-th copy $G_i$ induced by $\ALG(3n)$ for $1\leq i\leq n$.
    Whenever $c(T_i)\neq 0$, we have
    \begin{align}
        \label{eq:bound_density}
        \dens_G(T_i) \leq
        \begin{cases}
            \dens_G(A_i) = 1 & \text{if } A_i = T_i \\
            \dens_G(B_i) = \frac{6}{7} & \text{else.}
        \end{cases}
    \end{align}
    Since there are at most $k_n$ copies of $G_i$ for which $T_i=A_i$, we obtain
    \begin{align*}
        p(\ALG(3n)) &= \sum_{i\in \{1,\dots n\} : c(T_i)\neq0} \dens_G(T_i)c(T_i)\\
        &
        \leq \sum_{i\in \{1,\dots n\} : A_i = T_i} c(T_i) + \sum_{i\in \{1,\dots n\} : A_i \neq T_i} \frac{6}{7} c(T_i) \\
        &= \frac{6}{7} \sum_{i=1}^n c(T_i) + \frac{1}{7} \sum_{i\in \{1,\dots n\} : A_i = T_i} c(T_i)
        \leq \frac{18}{7} n + \frac{3}{7} k_n \\
        &
        \leq \frac{18\cdot17+30}{17\cdot7} n\\
        &= \frac{16}{17} \cdot 3n.
    \end{align*}
    For the optimal solution, we have $p(\OPT(3n-\alpha))>3n-\alpha-3$. For $n\geq\frac{\alpha}{3}+1$ and $B\coloneqq 3n-\alpha$ we obtain
    \begin{align*}
        \mu \geq \frac{p(\OPT(B))}{p(\ALG(B+\alpha))}
        > \frac{17}{16} \cdot \frac{3n-\alpha-3}{3n}.
    \end{align*}

    Next, assume $k_n> 10n/17$.
    We denote by $T_i=(V_{T_i},E_{T_i})\coloneqq G_i \cap \ALG(7n)$ the tree in the $i$-th copy~$G_i$ induced by $\ALG(7n)$ for $1\leq i\leq n$.
    For each copy $G_i$, we fix a subtree $T_i^*$ of $G_i$ that contains~$T_i$ as a subtree and collects the total prize of $G_i$, i.e., $T_i\subseteq T_i^*\subset G_i$ and $p(T_i^*)=6$.
    We can give a lower bound on the cost of $T^*_i$ by
    \begin{align} \label{eq:lower_bound_cost}
        c(T^*_i)\geq 
        \begin{cases}
            8 &  \text{if } a_i \in E_{T_i}, \\
            7 &  \text{else.}
        \end{cases}
    \end{align}
    Note that for all possible choices of $T_i$ and $T_i^*$, we have
    \begin{align} \label{eq:lower_bound_loss_density}
        p(T^*_i)-p(T_i) \geq \frac{3}{5} \bigl[c(T^*_i) -c(T_i) \bigr].
    \end{align}
    Since there are at least $k_n$ copies of $G_i$ for which $a_i\in E_{T_i}$, we obtain
    \begin{align*}
        p(\ALG(7n))
        &= p(G) - \sum_{i=1}^n \bigl[p(T^*_i) -p(T_i)\bigr] \\
        &
        \leq p(G) - \sum_{i=1}^n \frac{3}{5} \bigl[c(T^*_i)-c(T_i)\bigr] \\
        &
        \leq p(G) - \frac{3}{5} \left( \sum_{i\in\{1,\dots,n\}: a_i \in E_{T_i}} \bigl[8-c(T_i)\bigr] + \sum_{i\in\{1,\dots,n\}:\ a_i \notin T_i} \bigl[7-c(T_i)\bigr]\right) \\
        &= p(G) - \frac{3}{5} \left( \sum_{i\in\{1,\dots,n\}:\ a_i \in E_{T_i}} 1 + \sum_{i=1}^n \bigl[7-c(T_i) \bigr] \right) \\
        &\leq p(G) - \frac{3}{5} \left( k_n + 7n - 7n \right)\\
        &
        \leq 6n - \frac{30}{5 \cdot 17}n \\
        &= \frac{16}{17} \cdot 6n.
    \end{align*}
    For the optimal solution, we have $p(\OPT(7n-\alpha))>6n-\alpha-6$.
    Then, for $n\geq \frac{\alpha}{6}+1$ and $B\coloneqq 7n-\alpha$ we have
    \begin{align*}
        \mu \geq \frac{p(\OPT(B))}{p(\ALG(B+\alpha))}
        > \frac{17}{16} \cdot \frac{6n - \alpha -6}{6n}.
    \end{align*}

    Combining the two cases using $\frac{3n-\alpha-3}{3n}\leq \frac{6n - \alpha -6}{6n}$ we obtain $\mu \geq \frac{17}{16} \cdot \frac{6n - \alpha -6}{6n}$ for $n\geq \frac{\alpha}{3}+1$.
    Recall that $\alpha$ only depends on the maximum cost $\lrp=8$ of a vertex disjoint path, which is independent of the number $n$ of copies of $G_i$. Thus, we obtain 
    \begin{align*}
        \mu \geq \frac{17}{16} \liminf_{n\to \infty} \frac{6n - \alpha -6}{6n} = \frac{17}{16} = 1.0625
    \end{align*}
    This concludes the proof.
\end{proof}

\medskip
\noindent\textbf{Acknowledgements.}
We thank Felix Hommelsheim, Alexander Lindermayr, and Jens Schl\"oter for fruitful discussions and three anonymous referees for their comments that improved the presentation of the paper.

\bibliography{incremental}

\end{document}